\documentclass[10pt,journal,compsoc]{IEEEtran}

\usepackage[T1]{fontenc}
\usepackage{aecompl} 
\usepackage{caption}

\ifCLASSINFOpdf
   \usepackage[pdftex]{graphicx}

   \graphicspath{{{\vartheta}_{k}^v/pdf/}{{\vartheta}_{k}^v/jpeg/}}

   \DeclareGraphicsExtensions{.pdf,.jpeg,.png}
\else

   \usepackage[dvips]{graphicx}

   \graphicspath{{{\vartheta}_{k}^v/eps/}}
  
   \DeclareGraphicsExtensions{.eps}
\fi

\usepackage{amsmath}
\usepackage{subeqnarray}
\usepackage{cases}
\usepackage{cite}
\usepackage{graphicx}
\usepackage{subfigure}
\usepackage{tabularx,booktabs}
\usepackage{dingbat}
\usepackage{diagbox}
\usepackage[ruled,linesnumbered]{algorithm2e}
\usepackage{amssymb}
\usepackage{algorithmic}
\usepackage{esint}
\usepackage{color}
\usepackage{lipsum}
\usepackage{cuted}
\usepackage{stfloats}
\usepackage{multirow}
\usepackage{ntheorem}
\usepackage{color}
\usepackage{threeparttable}
\usepackage{bbm}
\usepackage{tabu}
\usepackage{bbding}
\usepackage{tikz}
\usepackage{xcolor}

\usepackage[colorlinks,
            linkcolor=black,       
            anchorcolor=black,  
            citecolor=black,        
            ]{hyperref}
\newcolumntype{C}{>{\centering\arraybackslash}X}
\theoremseparator{:}
\newtheorem{proposition}{\textcolor{black}{Proposition}}

\newtheorem{definition}{Definition}
\theoremstyle{nonumberplain}
\theorembodyfont{\normalfont}
\newtheorem{proof}{Proof}
\newtheorem{Proof}{Proof}

\definecolor{newextractedpurple}{RGB}{127,0,127}

\usepackage[cmintegrals]{newtxmath}

\begin{document}

\title{PACP: Priority-Aware Collaborative Perception for Connected and Autonomous Vehicles}
\vspace{-3cm}

\author{Zhengru~Fang,~
Senkang~Hu,
Haonan~An,
Yuang~Zhang,
Jingjing Wang, 
Hangcheng~Cao, 
Xianhao~Chen,~\IEEEmembership{Member,~IEEE}
and Yuguang Fang,~\IEEEmembership{Fellow,~IEEE}%
\IEEEcompsocitemizethanks{\IEEEcompsocthanksitem Z. Fang, S. Hu, H. An, H. Cao and Y. Fang are with the Department of Computer Science, City University of Hong Kong, Hong Kong. E-mail: \{zhefang4-c, senkang.forest, haonanan2-c\}@my.cityu.edu.hk, \{hangccao, my.fang\}@cityu.edu.hk.
\IEEEcompsocthanksitem Y. Zhang is with the Department of Civil and Environmental Engineering, University of Washington, Seattle, WA, USA. E-mail: yuangz19@uw.edu.
\IEEEcompsocthanksitem J. Wang is with the School of Cyber Science and Technology, Beihang University, China, and also with Hangzhou Innovation Institute, Beihang University, Hangzhou 310051, China. Email: drwangjj@buaa.edu.cn.
\IEEEcompsocthanksitem X. Chen is with the Department of Electrical and Electronic Engineering, the University of Hong Kong, Hong Kong. E-mail: xchen@eee.hku.hk (\textit{Corresponding author}).}
\thanks{This work was supported in part by the Hong Kong SAR Government under the Global STEM Professorship and Research Talent Hub, the Hong Kong Jockey Club under the Hong Kong JC STEM Lab of Smart City (Ref.: 2023-0108), and the Hong Kong Innovation and Technology Commission under InnoHK Project CIMDA. The work of Jingjing Wang was partly supported by the National Natural Science Foundation of China under Grant No. 62222101, Beijing Natural Science Foundation under Grant No. L232043 and No. L222039, and the Fundamental Research Funds for the Central Universities. The work of X. Chen was supported in part by HKU-SCF FinTech Academy R\&D Funding.}
}
\markboth{IEEE Transactions on Mobile Computing}%
{Shell \MakeLowercase{\textit{et al.}}: Bare Demo of IEEEtran.cls for IEEE Communications Society Journals}

\IEEEtitleabstractindextext{%
\begin{abstract}
  Surrounding perceptions are quintessential for safe driving for connected and autonomous vehicles (CAVs), where the Bird’s Eye View has been employed to accurately capture spatial relationships among vehicles. However, severe inherent limitations of BEV, like blind spots, have been identified. Collaborative perception has emerged as an effective solution to overcoming these limitations through data fusion from multiple views of surrounding vehicles. While most existing collaborative perception strategies adopt a fully connected graph predicated on fairness in transmissions, they often neglect the varying importance of individual vehicles due to channel variations and perception redundancy. To address these challenges, we propose a novel {\underline{P}}riority-{\underline{A}}ware {\underline{C}}ollaborative {\underline{P}}erception (\textbf{PACP}) framework to employ a BEV-match mechanism to determine the priority levels based on the correlation between nearby CAVs and the ego vehicle for perception. By leveraging submodular optimization, we find near-optimal transmission rates, link connectivity, and compression metrics. Moreover, we deploy a deep learning-based adaptive autoencoder to modulate the image reconstruction quality under dynamic channel conditions. Finally, we conduct extensive studies and demonstrate that our scheme significantly outperforms the state-of-the-art schemes by 8.27\% and 13.60\%, respectively, in terms of utility and precision of the Intersection over Union.
\end{abstract}
}

\maketitle

\section{Introduction}
\subsection{Background}
\IEEEPARstart{R}{ecent} advances in positioning and perception have emerged as pivotal components in numerous cutting-edge applications, most notably in autonomous driving\cite{hu2024magazine,hu2023adaptive,hu2023towards,chen2023vehicle,10415235}. These systems heavily rely on precise positioning and acute perception capabilities to safely and adeptly navigate complex road environments. Many solutions exist for positioning and perception in CAVs, including inertial navigation systems, high-precision GPS, cameras, and LiDAR\cite{9165167}. However, any isolated use of them may not be enough to achieve the desired level of perception quality for safe driving. In contrast, the Bird's Eye View (BEV) stands out as a more holistic approach. By integrating data from multiple sensors and cameras placed around a vehicle, BEV offers a comprehensive, potentially 360-degree view of a vehicle's surroundings, offering a more contextually rich understanding of its environment\cite{xu2022cobevt}. However, most BEV-aided perception designs have predominantly concentrated on single-vehicle systems. Such an approach may not be enough in high-density traffic scenarios, where unobservable blind spots caused by road obstacles or other vehicles remain a significant design challenge. Therefore, collaborative perception has become a promising candidate for autonomous driving. To mitigate the limitations of the single-vehicle systems, we can leverage multiple surrounding CAVs to obtain a more accurate BEV prediction via multi-sensor fusion\cite{zhang2024smartcooper}.

To further reduce the risks of blind spots in BEV prediction, collaborative perception is adapted to enable multiple vehicles to share more complementary surrounding information with each other through vehicle-to-vehicle (V2V) communications\cite{9682601}. This framework intrinsically surmounts several inherent constraints tied to single-agent perception, including occlusion and long-range detection limitation. Similar designs have been observed in a variety of practical scenarios, including communication-assisted autonomous driving within the realm of vehicle-to-everything\cite{9779322}, multiple aerial vehicles for accuracy perception\cite{liu2020when2com}, and multiple underwater vehicles deployed in search and tracking operations\cite{10193767,9451536,10517492}.

In this emerging field of autonomous driving, the current predominant challenge is how to make a trade-off between perception accuracy and communication resource allocation. Given the voluminous perception outputs (such as point clouds and consecutive RGB image sequences), the data transmission for CAVs demands substantial communication bandwidth. Such requirements often run into capacity bottleneck. As per the KITTI dataset\cite{6248074}, a single frame from 3-D Velodyne laser scanners encompasses approximately 100,000 data points, where the smallest recorded scenario has 114 frames, aggregating to an excess of 10 million data points. Thus, broadcasting such extensive data via V2V communications amongst a vast array of CAVs becomes a daunting task. Therefore, it is untenable to solely emphasize the perception efficiency enhancement without considering the overhead on V2V communications. Thereby, some existing studies propose various communication-efficient collaboration frameworks, such as the edge-aided perception framework\cite{chen2019f} and the transmission scheme for compressed deep feature map\cite{wang2020v2vnet}. It is observed that all these approaches can be viewed as fairness-based schemes, i.e., each vehicle within a certain range should have a fair chance to convey its perception results to the ego vehicle. Among all transmission strategies used for collaborative perception, the fairness-based scheme is the most popular one owing to its low computational complexity. 

Despite the low computational complexity, several major design challenges still exist with the state-of-the-art fairness-based schemes, entailing the adaptation of these perception approaches in real-world scenarios:
\begin{itemize}
  \item[1)] Fairness-based schemes may lead to overlapping data transmissions from multiple vehicles, causing unnecessary duplications and bandwidth wastage.
  \item[2)] Fairness-based schemes cannot inherently prioritize perception results from closer vehicles, which may be more critical than those from further vehicles.
  \item[3)] Without differentiating the priority of data from different vehicles, fairness-based schemes may block communication channels for more crucial information.
\end{itemize}

To address the above challenges, we conceive a priority-aware perception framework with the BEV-match mechanism, which acquires the correlation between nearby CAVs and the ego CAV. Compared with the fairness-based schemes, the proposed approach pays attention to the balance between the correlation and sufficient extra information. Therefore, this method not only optimizes perception efficiency by preventing repetitive data, but also enhances the robustness of perception under the limited bandwidth of V2V communications.

In the context of collaborative perception, a significant concern revolves efficient collaborative perception over time-invariant links. The earlier research on cooperative perception has considered lossy channels \cite{li2023learning} and link latency\cite{9718315}. However, most existing research leans heavily on the assumption of an ideal communication channel without packet loss, overlooking the effects of fluctuating network capacity\cite{9682601}. Moreover, establishing fusion links among nearby CAVs is a pivotal aspect that needs attention. Existing strategies are often based on basic proximity constructs, neglecting the dynamic characteristics of the wireless channel shared among CAVs. In contrast, graph-based optimization techniques provide a more adaptive approach, accounting for real-world factors like signal strength and bandwidth availability to improve network throughput.

A significant challenge in adopting these graph-based techniques is to manage the transmission load of V2V communications for point cloud and camera data. With the vast amount of data produced by CAVs, it is crucial to compress this data. Spatial redundancy is usually addressed by converting raw high-definition data into 2D matrix form. To tackle temporal redundancy, video-based compression techniques are applied. However, traditional compression techniques, such as JPEG\cite{skodras2001jpeg} and MPEG\cite{8945224}, are not always ideal for time-varying channels. This highlights the relevance of modulated autoencoders and data-driven learning-based compression methods that outperform traditional methods. These techniques excel at encoding important features while discarding less relevant ones. Moreover, the adoption of fine-tuning strategies improves the quality of reconstructed data, whereas classical techniques often face feasibility issues.

\subsection{State-of-the-Art}
\label{sec:Related Work}
In this subsection, we review the related literature of collaborative perception for autonomous driving with an emphasis on their perceptual algorithms, network optimization, and redundant information reduction. 

{\color{black}\textbf{V2V collaborative perception}: V2V collaborative perception combines the sensed data from different CAVs through fusion networks, thereby expanding the perception range of each CAV and mitigating troubling design problems like blind spots. For instance, Chen \textit{et al.} \cite{chen2019cooper} proposed the early fusion scheme, which fuses raw data from different CAVs, Wang \textit{et al.} \cite{wang2020v2vnet} employed intermediate fusion, fusing intermediate features from various CAVs, and Rawashdeh \textit{et al.} \cite{rawashdeh2018collaborative} utilized late fusion, combining detection outputs from different CAVs to accomplish collaborative perception tasks. Although these methods show promising results under ideal conditions, in real-world environments, where the channel conditions are highly variable, directly applying the same fusion methods often results in unsatisfactory outcomes. Transmitting raw-level sensing data has distinct advantages: it supports multiple downstream tasks, enhances data fusion accuracy, and ensures future-proof flexibility for evolving autonomous driving technologies. Therefore, we can combine the benefits of early fusion and intermediate fusion by utilizing adaptive compression for raw data transmission.}

\textbf{Network optimization}: High throughput can ensure more efficient data transmissions among CAVs, thereby potentially improving the IoU of cooperative perception systems. Lyu \textit{et al.}\cite{lyu2022distributed} proposed a fully distributed graph-based throughput optimization framework by leveraging submodular optimization. Nguyen \textit{et al.}\cite{9204672} designed a cooperative technique, aiming to enhance data transmission reliability and improve throughput by successively selecting relay vehicles from the rear to follow the preceding vehicles. Ma \textit{et al.}\cite{8812911} developed an efficient scheme for the throughput optimization problem in the context of highly dynamic user requests. However, the intricate relationship between throughput maximization and IoU has not been thoroughly investigated in the literature. This gap in the research motivates us to conduct more comprehensive studies on the role of throughput optimization in V2V cooperative perception. 

\textbf{Vehicular data compression}: For V2V collaborative perception, participating vehicles compress their data before transmitting it to the ego vehicle to reduce transmission latency. However, existing collaborative frameworks often employ very simple compressors, such as the naive encoder consisting of only one convolutional layer used in V2VNet\cite{wang2020v2vnet}. Such compressors cannot meet the requirement of transmission latency under $100$ ms in practical collaborative tasks\cite{zhang2018vehicular}. Additionally, current views suggest that compressors composed of neural networks outperform the compressors based on traditional algorithms\cite{yang2020variable}. However, these studies are typically focused on general data compression tasks and lack research on adaptive compressors suitable for practical scenarios in V2V collaborative perception. 

{\color{black}
\textbf{Priority-Aware Perception Schemes}: Priority-aware perception schemes have been advanced significantly, yet they have also encountered limitations in dynamic environments. Liao \textit{et al.}'s model uses an attention mechanism for trajectory prediction, effectively assigning dynamic weights but overlooking communication costs and information redundancy\cite{liao2024bat}. Similarly, Wen \textit{et al.}'s queue-based traffic scheduling enhances throughput but struggles with real-time data synchronization, which is critical for systems like autonomous vehicles\cite{10262233}. Additionally, studies on edge-assisted visual simultaneous localization and mapping (SLAM) introduce a task scheduler that improves mapping precision but fails to account for the crucial role of channel quality in perception accuracy and throughput\cite{10327699}.

Our framework mitigates these issues by introducing an adaptive mechanism that adjusts priorities based on real-time channel quality and analytics. This approach not only curtails unnecessary data processing and communication but also enhances raw-level sensing data fusion and system responsiveness, thereby addressing the core challenges identified in previous studies.
}

\subsection{Our Contributions}
\begin{table*}[!t]
  \scriptsize
  \caption{CONTRASTING OUR CONTRIBUTION TO THE LITERATURE}
  \renewcommand{\arraystretch}{1}
  \centering
\begin{tabular}{|l|c|c|c|c|c|c|c|c|c|c|c|c|c|c|c|c|}
\hline
& \cite{an2024throughput} & \cite{hu2023towards} & \cite{hu2023adaptive}& \cite{li2023learning}& \cite{liu2020when2com}  & \cite{chen2019f}& \cite{wang2020v2vnet} &\cite{chen2019cooper}& \cite{rawashdeh2018collaborative} & \cite{lyu2022distributed} & \cite{9204672} & \cite{8812911} & \cite{9681261} & Proposed work \\
\hline
\hline
BEV evaluation & \checkmark & \checkmark & \checkmark & \checkmark &  &  & \checkmark &  &  &  &  &  &  & \checkmark \\
\hline
Multi-agent selection & \checkmark &  & \checkmark &  & \checkmark &\checkmark &\checkmark  & \checkmark &  &\checkmark  &\checkmark  &  & \checkmark & \checkmark \\
\hline
Data compression & \checkmark &  & \checkmark &  & \checkmark& \checkmark &\checkmark  & \checkmark &  &  &  &  &  & \checkmark \\
\hline
Lossy communications & \checkmark &  & \checkmark & \checkmark &  & &  &  &\checkmark  & \checkmark  &\checkmark  &  & \checkmark & \checkmark \\
\hline
Priority mechanism &  &  &  &  &  & &  &  &  &  &  &  & \checkmark & \checkmark \\
\hline
Coverage optimization&  &  &  &  & \checkmark & &  &  &  &  &\checkmark  &  & \checkmark & \checkmark \\
\hline
Throughput optimization & \checkmark &  & \checkmark &  &  & &  &  &  & \checkmark & \checkmark & \checkmark & \checkmark & \checkmark \\
\hline
\end{tabular}\label{table_literature}
\end{table*}

To address the weakness of prior works and tackle the aforementioned design challenges, we design our {\underline{P}}riority-{\underline{A}}ware {\underline{C}}ollaborative {\underline{P}}erception (\textbf{PACP})
framework for CAVs and evaluate its performance on a CAV simulation platform CARLA\cite{dosovitskiy2017carla} with OPV2V dataset\cite{xu2022opv2v}. Experimental results verify PACP's superior performance, with notable improvements in utility value and the average precision of Intersection over Union (AP@IoU) compared with existing methods. To summarize, in this paper, we have made the following major contributions:
{\color{black}
\begin{itemize}
  \item We introduce the first-ever implementation of a priority-aware collaborative perception framework that incorporates a novel BEV-match mechanism for autonomous driving. This mechanism uniquely balances communication overhead with enhanced perception accuracy, directly addressing the inefficiencies found in prior works.
  \item Our two-stage optimization framework is the first to apply submodular theory in this context, allowing the joint optimization of transmission rates, link connectivity, and compression ratio. This innovation is particularly adept at overcoming the challenges of data-intensive transmissions under dynamic and constrained channel capacities.
  \item We have integrated a deep learning-based adaptive autoencoder into PACP, supported by a new fine-tuning mechanism at roadside units (RSUs). The experimental evaluation reveals that this approach surpasses the state-of-the-art methods, especially in utility value and AP@IoU.
\end{itemize}
}

Our new contributions are better illustrated in Table \ref{table_literature}.

\section{Fairness or Priority?}
\label{sec:Fairness or Priority}
In this section, we show some practical difficulties with fairness-based collaborative perception for CAVs and also briefly demonstrate the superiority of adopting a priority-aware perception framework.

\subsection{Background of Two Schemes}
\label{sec:Background of Two Schemes}
{\color{black}
\textbf{Fairness-based scheme}: This scheme aims to achieve fairness in resource allocation among different CAVs. The Jain's fairness index is used to measure fairness:\cite{9681261}:
\begin{equation}\label{Jain's Fairness Index}
     J = \frac{(\sum_{i=1}^{n} x_i)^2}{n \cdot \sum_{i=1}^{n} x_i^2},
\end{equation}
where \( n \) is the total number of nodes and \( x_i \) is the resource allocated to the \( i \)-th node. A perfect fairness index of 1 indicates equal resource allocation. Two common fairness schemes are subchannel-fairness (equal spectrum resources) and throughput-fairness (equal transmission rates).

\textbf{Priority-aware scheme}: Unlike fairness-based schemes, this scheme assigns different priority levels to CAVs based on the importance and quality of their data. The ego vehicle gives higher priority to CAVs with better channel conditions and more crucial perception data. Existing works have investigated several popular priority factors as priority, such as link latency\cite{9403386} and routing situation\cite{luo2022learning}. This approach mitigates the negative impact of "poisonous" CAVs with poor channel conditions or less relevant data, thus enhancing the overall system efficiency. 
Compared to the fairness-based scheme, the advantages of priority-aware perception can be summarized as follows:
\begin{itemize}
\item \textbf{Transmission Efficiency}: By prioritizing CAVs with better data quality and channel conditions, the priority-aware scheme optimally allocates spectrum resources, ensuring efficient transmission.
\item \textbf{Improved Prediction}: This scheme reduces the influence of poisonous data, leading to more accurate BEV predictions by focusing on high-quality perception inputs.
\item \textbf{Dynamic Adaptability}: The priority-aware scheme dynamically adjusts to changing channel conditions and task requirements, maintaining robust performance in diverse environments.
\end{itemize}

\subsection{An Illustrative Motivating Example}
\label{sec:Inefficient Bandwidth Allocation}
\begin{figure*}
    \centering
    \subfigure[Subchannel-fairness scheme]{
        \includegraphics[width=5.4cm]{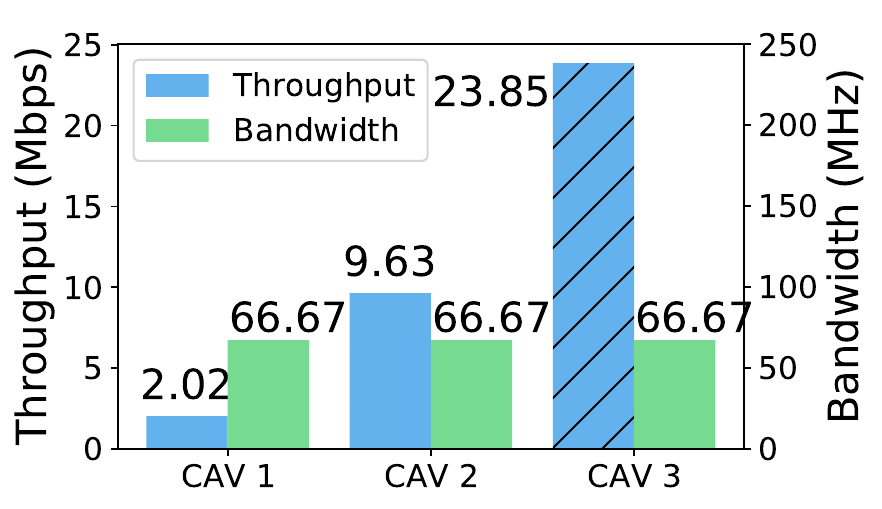}\label{fig:DMDDAC}
    }
    \subfigure[Throughput-fairness scheme]{
        \includegraphics[width=5.4cm]{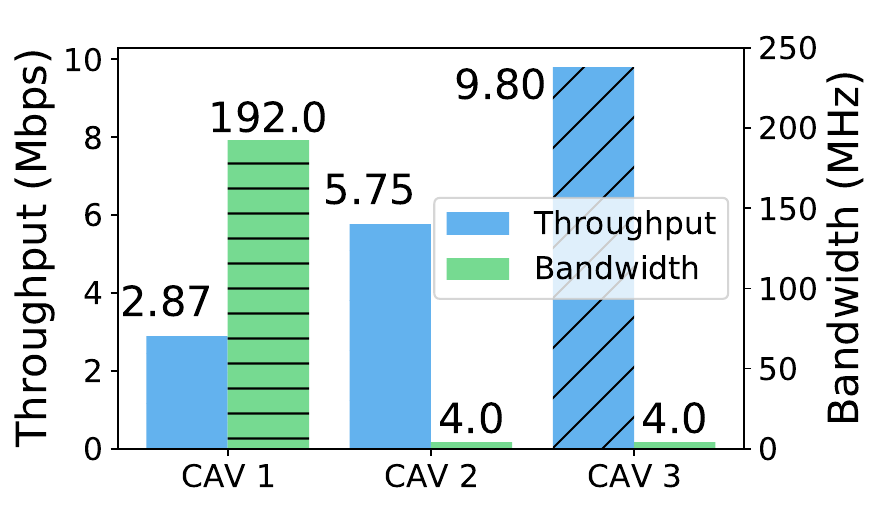}\label{fig:FTS}
    }
    \subfigure[Priority-aware perception scheme]{
        \includegraphics[width=5.4cm]{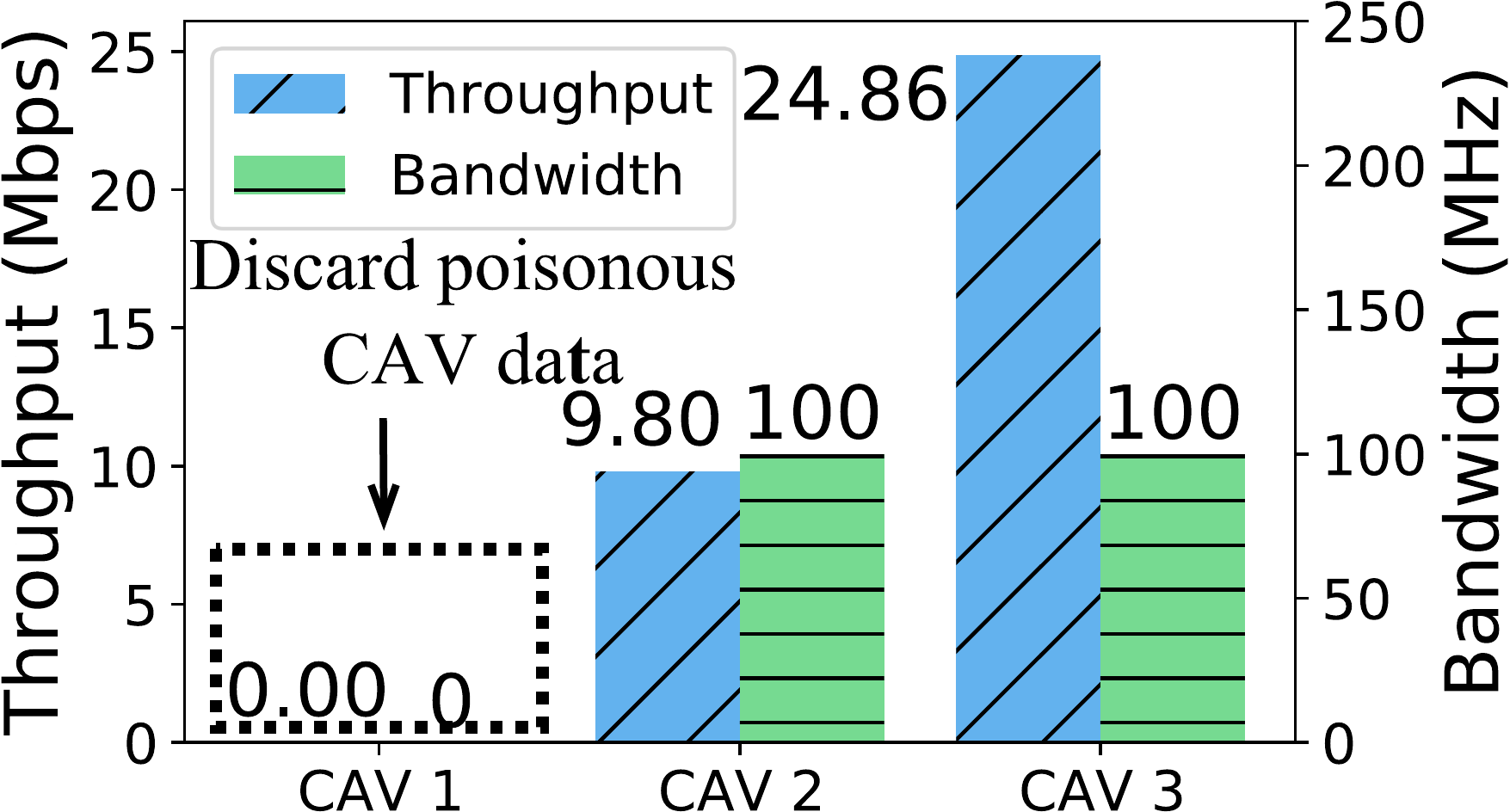}\label{fig:PACP}
    }
    \caption{The bandwidth and throughput allocation by different schemes within V2V network.}
    \label{fig: pre_subfigures}
\end{figure*}
\begin{figure}[t]
  \centering
  \subfigure[Camera perception by ego CAV]{
  \includegraphics[width=3.7 cm]{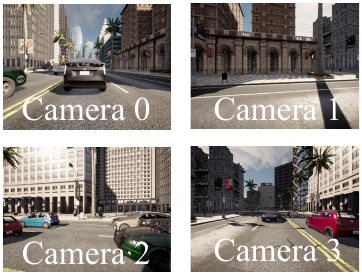}\label{fig:camera}
  }
  \subfigure[AP@IoU vs. Different schemes]{
  \includegraphics[width=4.5 cm]{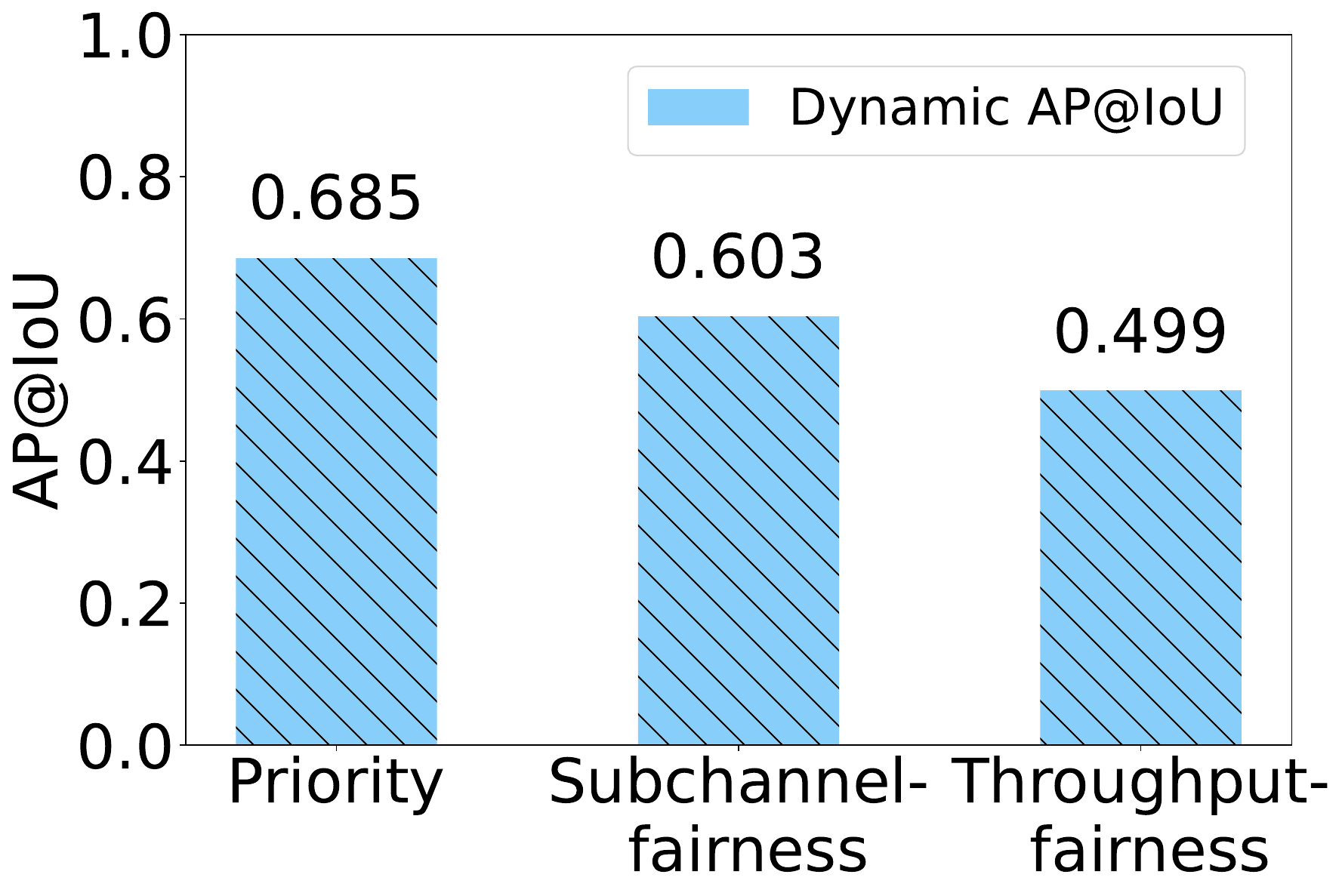}\label{fig:pre_iou}
  }
  \caption{Camera data and different types of AP@IoU.}
  \label{fig:preliminary_result_bev}
\end{figure}
To understand the limitations of fairness-based schemes, we compare their resource allocation and BEV prediction against an ad hoc priority-aware scheme. Fig. \ref{fig: pre_subfigures} illustrates the bandwidth and throughput allocation of three schemes in a V2V network, with a total bandwidth of 200 MHz.

In the subchannel-fairness scheme (Fig. \ref{fig:DMDDAC}), each CAV receives an equal amount of spectrum resources, leading to inefficient utilization, especially with poor channel conditions. The throughput-fairness scheme (Fig. \ref{fig:FTS}) equalizes transmission rates by allocating more resources to weaker channels, but it can still result in suboptimal BEV predictions. For instance, if a CAV has extremely poor channel quality (e.g., CAV 1), even the allocated bandwidth may be insufficient, reducing perception performance.

In contrast, the priority-aware perception scheme (Fig. \ref{fig:PACP}) dynamically adjusts the priority weights of each CAV based on channel resources, perception accuracy, and coverage. For example, CAV 1, with the worst channel condition, is assigned with the lowest priority and its data is discarded. This approach ensures that critical and high-quality data from CAVs with the best channel conditions and/or proximity to the ego vehicle are transmitted with minimal loss and latency, significantly improving BEV predictions. Fig. \ref{fig:pre_iou} demonstrates the superior performance of the priority-aware scheme, with the dynamic AP@IoU metric rising to 0.685. Thus, this ad hoc priority-aware scheme outperforms fairness-based methods by selectively collaborating with the most valuable CAVs, leading to enhanced BEV predictions. However, how to design an effective priority-aware scheme is critical, which motivates this research. 
}

\section{System Model}
\label{sec:System_Model}

In this section, we present a V2X-aided collaborative perception system with CAVs, including the system's structure, channel modeling, and the constraints of computational capacity and energy. The key notations are listed in the \mbox{Table \ref{table:notation}}.
\subsection{System Overview}
\label{sec:System Overview}
We consider a V2X-aided collaborative perception system with multiple CAVs and RSUs, which is shown in Fig. \ref{fig:system model}. In our scenario, CAVs can be divided into two types. The first type is the nearby CAVs (indexes 1-3), which monitor the surrounding traffic with cameras and share their perception results with other CAVs. The second type is the ego CAV (index $0$), which fuses the camera data from the nearby CAVs with its own perception results. As shown in Fig. \ref{fig:system model}, the ego CAV $0$ is covered with the parked car that the ego CAV’s own camera cannot observe the incoming pedestrian from the blind spot. Through the collaborative perception scheme, the ego CAV merges compressed data from CAVs 2 and 3, i.e., $s_{20}=s_{30}=1$, which catch the existence of the pedestrian. However, it is unnecessary to connect all CAVs together since the bandwidth and subchannel resources are limited. Therefore, the ego CAV can determine the importance of each nearby CAV by the priority-aware mechanism, depending on CAVs' positions and channel states. For example, the ego CAV is disconnected from CAV 1, because it fails to provide enough environmental perception information, i.e., $s_{10}=0$. For the sake of maximizing AP@IoU, the ego CAV obtains the near-optimal solution in terms of the transmission rate $d_{ij}$ and the compression ratio $r_{ij}$ at each time slot. Additionally, we deploy several RSUs to achieve a fine-tuning compression strategy, which is detailed in Sec. \ref{sec:Deep Learning-Based Adaptive Compression}.
\begin{figure}[t]
  \centering
  \includegraphics[width=0.47\textwidth]{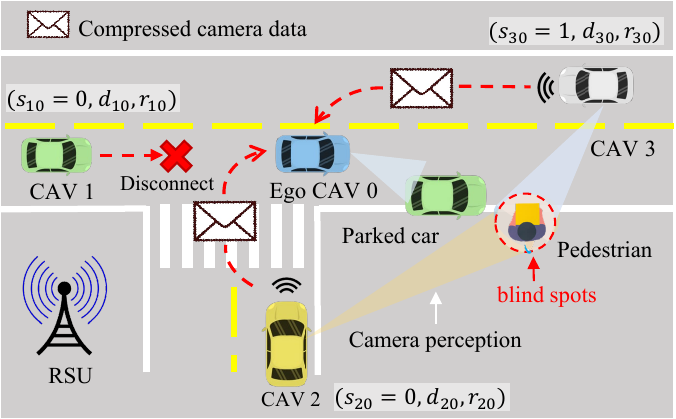}
  \caption{Overview of V2X-aided collaborative perception system. }
  \label{fig:system model}
\end{figure}

\subsection{Channel Modeling and System Constraints}
\label{sec:Channel Modeling}
\begin{table}[t]
  \caption{{\color{black}Summary of Key Notations}}
  \label{table:notation}
  \begin{center}
  \renewcommand{\arraystretch}{1.5}
  \begin{tabular}{c|p{6.5cm}}
  \hline
  \textbf{Notation} & \textbf{Definition} \\
  \hline
  \( N \) & The total number of vehicles in the network \\
  \( K \) & The number of orthogonal sub-channels for the whole V2V networks \\
  \( W \) & The total bandwidth of the V2V networks \\
  \( C_{ij} \) & The channel capacity between the \(i\)th transmitter and the \(j\)th receiver \\
  \( \mathcal{D} \) & The matrix of transmission rates \\
  \( r_{ij} \) & The adaptive compression ratio \\
  \( E_{ij}^t \) & The energy bounds for data transmission \\
  \( E_{j}^c \) & The energy consumption bounds for computation \\
  \( E_j^T \) & The energy consumption threshold \\
  \( P_t \) & Transmission power of each CAV \\
  \( F_j \) & CPU capacity of CAV \(j\) \\
  \( \beta \) & Model complexity parameter depending on the architecture of the neural networks \\
  \( A_j \) & Local data generation rate per second at CAV \(j\) \\
  \( \mathcal{P}_{ij} \) & The priority weight between two CAVs \\
  \( \mathcal{I}(\mathcal{S}_i) \) & The total area covered by the union of perceptual regions of vehicles \\
  \( \mathcal{U}_{\mathrm{sum}} \) & The weighted utility function for the network  \\
  \( \mathcal{U}_{\mathrm{r}} \) & The sub-utility function for perception quality \\
  \( \mathcal{U}_{\mathrm{p}} \) & The sub-utility function for the perceptual regions \\
  \( (\omega_1, \omega_2) \) & Weights for perception quality and region \\
  \( (r_{\min}, r_{\max}) \) & Compression ratio range \\
  \hline
  \end{tabular}
  \end{center}
\end{table}
Consider the V2V network architecture $\mathbf{G} = (\mathbf{V}, \mathbf{E})$, where $\mathbf{V} = ({v_1, v_2, ..., v_N})$ denotes CAVs and $\mathbf{E}$ is the set of links between them. As per 3GPP specifications for 5G\cite{8891313}, V2V networks adopt Cellular Vehicle-to-Everything (C-V2X) with Orthogonal Frequency Division Multiplexing (OFDM). The total bandwidth $W$ is split into $K$ orthogonal sub-channels. Each sub-channel capacity is $C_{i j}=\frac{W}{K} \log_{2}\left(1+\frac{P_{t} h_{i j}}{N_{0}\frac{W}{K}}\right)$, where $P_t$ represents the transmit power, $h_{ij}$ denotes the channel gain from the $i$th transmitter to the $j$th receiver, and $N_0$ is the noise power spectral density. 

Additionally, let \(s_{ij} = 1\) indicate the presence of the directional link from CAV \(i\) to ego CAV \(j\). Such a link is denoted as \((i, j) \in \mathcal{S}\). The set \(\mathcal{S}\) represents the collection of all established links in the network. When \(s_{ij} = 1\), CAV \(v_{i}\) is capable of sharing data with the ego CAV \(v_{j}\). Conversely, if \(s_{ij} = 0\), it implies the disconnected mode of the link \((i, j)\). However, the number of directed links potentially increase at a rate of \(N^2\) with the number of CAVs, possibly exhausting the limited communication spectrum resources. Therefore, the upper bound of the number of connections is given by:
\begin{equation}\label{eq:subchannel}
  \begin{aligned}
    \sum_{i=1, i\ne j}^N{\sum_{j=1}^N{s_{ij}}}\le K.
  \end{aligned}
\end{equation}
Let $\mathcal{D}=\left\{d_{ij} \right\}_{N \times N}$ be the matrix consisting of  transmission rates, where each element is non-negative, for $\forall (i, j) \in \mathbf{E}$. Each element $d_{ij}$ represents the data rate without compression from vehicle $v_i$ to vehicle $v_j$, which is then processed by $v_j$. It is noteworthy that $d_{ij}$ satisfies:
\begin{equation}\label{eq:transmission rate}
\begin{aligned}
{r}_{ij}d_{ij} \le \min(C_{ij}, {r}_{ij}A_i),
\end{aligned}
\end{equation}
Here, ${r}_{ij} \in (0,1]$ denotes the adaptive compression ratio, obtained by the compression algorithm outlined in \mbox{Sec. \ref{sec:Deep Learning-Based Adaptive Compression}}. $r_{ij} d_{ij}$ represents the actual transmission rate after compression. $A_i$ signifies the amount of local perception data at $v_i$ per second, i.e., perception data generation rate at the location of $v_i$. This constraint implies that the actual transmission rate must be limited either by the achievable data rate or by the locally compressed data present at vehicle $v_i$. Furthermore, it is observed that an inadequate compression ratio results in a diminution in the accuracy of perception data, while an excessively high compression ratio results in suboptimal throughput maximization. Consequently, the constraint for the compression ratio is defined as follows:
\begin{equation}\label{eq:the constraint of the compression ratio}
  \begin{aligned}
    \mathbf{1}^{\top}{r} _{j,\min}\preceq \mathcal{R} _j\preceq \mathbf{1}^{\top}{r} _{j,\max},
\end{aligned}
\end{equation}
where $ \mathcal{R} _j=\left[ \mathrm{{r}}_{{1j}},\mathrm{{r}}_{2{j}},...,\mathrm{{r}}_{{Nj}} \right] ^{\top}$, \textbf{${\mathcal{R}}$} $=\left[ \mathcal{R} _1,\mathcal{R} _2,...,\mathcal{R} _{{N}} \right]$. 
Given the surrounding data obtained through collaborative perception, perception data from closer vehicles are more important for perceptional detection, which has a higher level of accuracy. Therefore, we assume that the adaptive compression ratio for the link $(i,j)$ yields:
\begin{equation}\label{eq:the constraint of the compression ratio 2}
  \begin{aligned}
    {r} _{ij}e^{L_{ij}}\geqslant \eta,
\end{aligned}
\end{equation}
where $L_{ij}$ denotes the normalized distance between $v_i$ and $v_j$, and $\eta \in (0,1]$. It is noted that we use an exponential relationship in terms of normalized distance, because the compression ratio of the remote area should decrease rapidly, reducing communication overhead.
Moreover, the link establishment and data transmission rate should satisfy the bounds of energy consumption as follows:
\begin{equation}\label{eq:transmission power}
  \begin{aligned}
    E_{ij}^{t}=\tau _{j}^{t}P_t s_{ij}, 
  \end{aligned}
\end{equation}
where $\tau_{j}^{t}$ denotes the allocated time span, and $P_{t}$ signifies the transmission power. We define $F_{j}$ as the computational capability of vehicle $v_{j}$. The data processed by $v_{j}$, which includes its local data $A_{j}$ and the data received from neighboring nodes, should satisfy the following constraint:
\begin{equation}\label{eq:computing power constraint}
  \begin{aligned}
    A_j+\sum_{i=1,i\ne j}^N{{r} _{ij}s_{ij}d_{ij}}\leqslant F_j/\beta, 
  \end{aligned}
\end{equation}
where $F_j/\beta$ represents the aggregate size of data processed per second. Additionally, $\beta$ is tunable parameter depending on the architecture of the neural networks employed in these contexts, like the self-supervised autoencoder. The energy consumption for computation by $v_j$ can be determined as follows:
\begin{equation}\label{eq:computing power}
  \begin{aligned}
    E_{j}^{c}=\left( A_j+\sum_{i=1,i\ne j}^N{{r} _{ij}s_{ij}d_{ij}} \right)\epsilon _j\tau_{j}^{c},
  \end{aligned}
\end{equation}
where $\epsilon_{j}$ denotes the energy cost per unit of input data processed by $v_{j}$'s processing unit. $\tau_{j}^{c}$ represents the duration allocated for data processing. By imposing constraints on the overall energy consumption, a suitable trade-off between computing and communication can be made, facilitating optimal operation and extending the operational longevity of CAVs. Intuitively, the cumulative energy consumption in our CAV system must satisfy the following constraint:
\begin{equation}\label{eq:power constraint}
  \begin{aligned}
    \sum_{i=1,i\ne j}^N \left(E_{ij}^{t}+E_{ij}^{c} \right) \leq E^{T}_j, \quad (j=1,2,\cdots,N)
  \end{aligned},
\end{equation}
where $E^{T}_j$ symbolizes the energy consumption threshold for the $j$th CAV group, including its nearby CAVs.

\section{Priority-Aware Collaborative Perception Architecture}
\label{sec:Priority-Aware Perception Scheme}
While BEV offers a top-down view aiding CAVs in learning relative positioning, not all data is of equal importance or relevance\cite{9779322}. Some CAV data can be unreliable due to perception qualities, necessitating differential prioritization in data fusion. This section proposes a priority-aware perception scheme by BEV match mechanism.
\subsection{Selection of Priority Weights}
\label{sec:priority_weights}
There exist many ways to define priority weights, such as distance\cite{adil2021enhanced}, channel state\cite{9849675}, and information redundancy\cite{10158439}. However, as CoBEVT is the backbone for RGB data fusion\cite{xu2022cobevt}, our priority weight definition is based on the Intersection over Union (IoU) of BEV features between adjacent CAVs. IoU is an effective metric in computer vision, especially in object detection, used to quantify the overlap between two areas. The background of IoU can be obtained in CoBEVT's local attention aids in pixel-to-pixel correspondence during object detection fusion. IoU of the BEV features reveals the consistency in environmental perception between CAVs. Factors like channel interference and network congestion may cause inconsistencies. As consistency is crucial in fusion data, inconsistencies can lead to data misrepresentations. Hence, we design a BEV-match mechanism by relying on IoU analysis in the next subsection, giving preference to CAVs with closely aligned perceptions to the ego CAV.

\subsection{Procedure of Obtaining Priority Weights}
\label{sec:obtaining_priority_weights}
\begin{figure*}[t]
  \centering
  \includegraphics[width=0.85\textwidth]{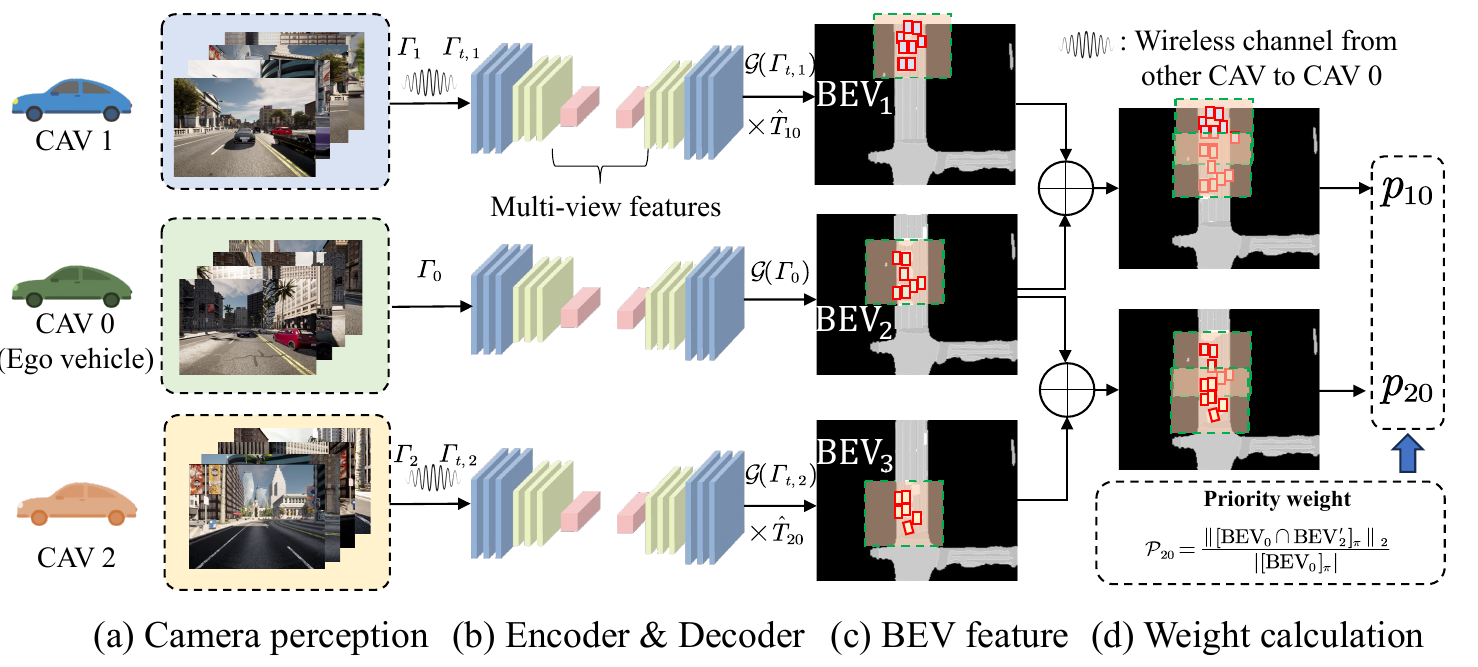}
  \caption{Procedure for priority weight calculation. Fig. \ref{fig:pacs}(a): CAVs observe surroundings with 4 cameras; CAVs 1-2 relay RGB data to CAV 0. Figs. \ref{fig:pacs}(b)-(c): BEV feature generation in CAV 0. Fig. \ref{fig:pacs}(d): BEV-match mechanism determines priority weights.}
  \label{fig:pacs}
  \vspace{-4mm}
\end{figure*}

The procedure of calculating priority weights can be divided into three steps as follows:

\textbf{(1) Camera perception:} As shown in Fig. \ref{fig:pacs}(a), the nearby CAVs capture raw camera data \( \varGamma \) using its four cameras. This perception data \( \varGamma   \) is then transmitted to the ego CAV by communication units through wireless channel. Let $\varGamma _t = \mathcal{F}(\varGamma )$ denote the perception data received by the ego CAV, where \( \mathcal{F} \) represents the function of data transmission over the network. In reality, different data processing strategies result in different impacts on $\varGamma _t$, such as compression and transmission latency. 

\textbf{(2) Encoder \& Decoder:} Fig. \ref{fig:pacs}(b) illustrates that upon reception, the ego CAV uses a SinBEVT-based neural network to process the single-vehicle's RGB data\cite{xu2022cobevt}. This data is transformed through an encoding-decoding process to extract BEV features. The BEV feature transformation is denoted by $\text{BEV} = \mathcal{G}(\varGamma _t)$, where \( \mathcal{G} \) captures SinBEVT's processing essence. In Fig. \ref{fig:pacs}(c), the BEV feature depicts the traffic scenario from a single vehicle's perspective, with the green dashed box marking the range of perceived moving vehicles, and the red box indicating surrounding vehicles to the ego CAV. In fact, SinBEVT and CoBEVT achieve real-time performance of over 70 fps with five CAVs\cite{xu2022cobevt}.

\textbf{(3) Priority weight calculation:} As shown in Fig. \ref{fig:pacs}(d), we take the ego CAV and CAV 1 as an example. Along with their corresponding BEV perceptions \( \text{BEV}_0 \) and \( \text{BEV}_1 \), let the coordinates and orientation angle of both the Ego CAV and CAV 1 be represented as \( (x_0,y_0,\theta_0) \) and \( (x_1,y_1,\theta_1) \), respectively. Firstly, we can derive the translational displacement between Ego CAV and CAV 1: $\Delta x = x_0 - x_1 $, $\Delta y = y_0 - y_1$ and $\Delta \theta = \theta_0 - \theta_1$. Accordingly, the translation matrix \( T \) and the rotation matrix $\varTheta$\footnote{We assume that the rotation matrix provided is based on the counter-clockwise rotation convention.} are articulated as:
\begin{equation}\label{eq:translation_matrix}
T_{10}=\left[ \begin{matrix}
	1&		0&		\Delta x\\
	0&		1&		\Delta y\\
	0&		0&		1\\
\end{matrix} \right] , \varTheta_{10}=\left[ \begin{matrix}
	\cos\mathrm{(}\Delta \theta )&		-\sin\mathrm{(}\Delta \theta )&		0\\
	\sin\mathrm{(}\Delta \theta )&		\cos\mathrm{(}\Delta \theta )&		0\\
	0&		0&		1\\
\end{matrix} \right] ,
\end{equation}
where $\hat{T}_{10} = T_{10} \times \varTheta_{10}$ remap coordinates from \( \text{BEV}_1 \) to \( \text{BEV}_0 \). Let \( \hat{\varGamma}  \) be a point of the BEV feature of CAV 1. Furthermore, the transformation of a point \( \hat{\varGamma}  \) from the BEV perception \( \text{BEV}_1 \) to \( \text{BEV}_0 \) is denoted by \( \hat{\varGamma} ' = \hat{T}_{10} \times \hat{\varGamma}  \). \( \hat{\varGamma} ' \) represents the transformed coordinate in \( \text{BEV}_0 \) and the point \( \hat{\varGamma}  \) is articulated in homogeneous coordinates, i.e., \( \hat{\varGamma}  = [x, y, 1]^T \). 

{\color{black}
Inspired by IoU, the metric for perceptual quality depends on the intersection of the ground truth and other predicted results. As we are mainly concerned with the impact of an unstable channel on perceptual quality, we can assume that the ego vehicle ($\text{CAV 0}$ in Fig. \ref{fig:priority}) can obtain highly accurate locations of nearby objects, which can serve as the ground truth. For each nearby CAV within communication range, $\text{CAV 0}$ can calculate priority weights only using overlapping perceptual objects, i.e., we only calculate priority weights based on BEV features (boxes) about objects $\pi_0,\pi_1,\pi_2,$ and $\pi_3$. Specifically, the priority weight $\mathcal{P}_{10}$ is formulated as follows:
\begin{equation}\label{eq:bev iou}
  \begin{aligned}
     \mathcal{P}_{10} =\frac{\| [ \mathrm{BEV}_0 \cap \mathrm{BEV}_{1}^{\prime} ]_\pi\|_2}{| [ \mathrm{BEV}_0 ]_\pi |} ,
  \end{aligned}
\end{equation}
where \( \text{BEV}_{\text{0}} \) denote the ego CAV's BEV features and \( \text{BEV}_1^{\prime}=\hat{T}_{10}\times \text{BEV}_1 \) signifies the transformed BEV, the operation \( [ \cdot ]_\pi \) zeroes out pixels outside the overlapping perceptual objects of the fused BEV's non-zero region, \( \| \cdot \|_2 \) measures the square root of squared pixel values, \( | \cdot | \) sums all pixel values in an image, and the numerator in Eq. (\ref{eq:bev iou}) quantifies perception difference between data from CAVs 1 and 0, with the denominator providing normalization. The term $\mathcal{P}$ in fact captures two characteristics between the view of the ego CAV and the transformed view of the assisting CAV for data fusion. If two views are very close, the intersection will be similar, the numerator in Eq. (\ref{eq:bev iou}) will be very high, thus ${\mathcal{P}}_{10}$ will be very high, which implies that the view from CAV 1 will accurate enough to enhance the perception  quality, meaning that the view from CAV 1 can be assigned high priority. If the view from CAV 1 is too dissimilar, the intersection will be close to empty, implying that the numerator in Eq. (\ref{eq:bev iou}) will be close to 0, meaning that the view from CAV 1 may be more misleading and should be assigned lower priority.
\begin{figure}[t]
  \centering
  \includegraphics[width=0.47\textwidth]{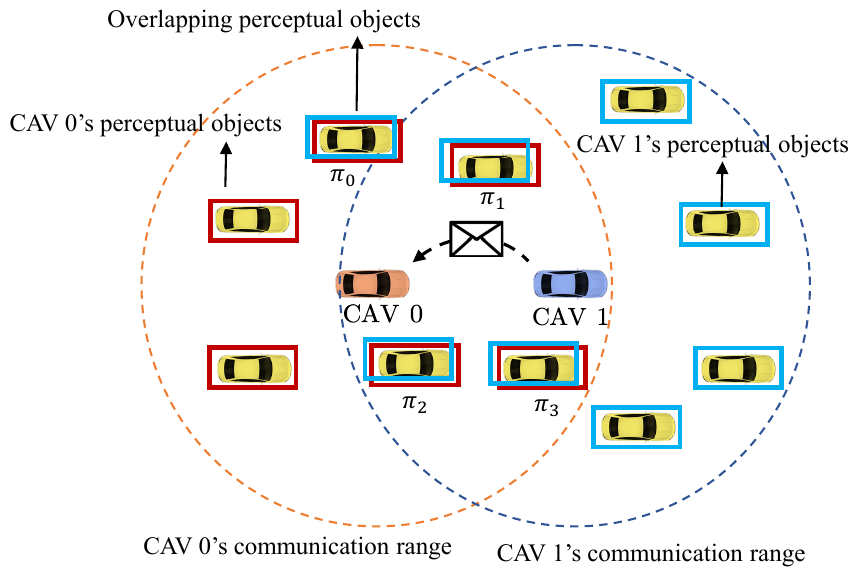}
  \caption{{\color{black}An example of determining priority weight $\mathcal{P}_{10}$.} }
  \label{fig:priority}
\end{figure}

Moreover, our system incorporates a gate mechanism that effectively filters out data from CAVs whose priority weights fall below a certain threshold. This gating process prevents low-quality data from influencing the overall perception quality, thereby maintaining the integrity and accuracy of the collaborative perception system despite the inherent instabilities in vehicular network conditions.
}
\subsection{Utilization Maximization in Collaborative Perception}
\label{sec:Utilization maximization in collaborative perception}
The performance of collaborative perception relies on the quality of the V2V communication network. Hence, our aim is to maximize the amount of perception data transmitted under the constraints of communication resources, such as power and bandwidth. However, given the varying contributions of nearby CAVs, it is unwise to directly optimize the total throughput of the V2V networks. Instead, we employ the aforementioned priority weight to adjust the resources allocated to each CAV. 
{\color{black}Moreover, one group of CAVs might have overlapping perceptual regions. To handle this, we leverage the union-style perception region set $G_i$, defined as:$G_i = \bigcup_{j \in J_i} \mathcal{G}_j $, where $\mathcal{G}_j$ represents the perception region of the $j$th CAV, and $J_i = \{j | s_{ji} = 1\}$. Here, $\mathcal{I}(\mathcal{S}_i)$ represents the total area covered by the union of perceptual regions of vehicles in set $J_i$, which is mathematically expressed as:
\begin{equation}\label{eq:union}
\mathcal{I}(\mathcal{S}_i) = \mathcal{A} \left( \bigcup_{j \in J_i} \mathcal{G}_j \right)
\end{equation}
where $\mathcal{A}(X)$ denotes the area of a region $X$. The utilization of the union of regions inherently introduces a submodular property in the coverage function. Submodularity\footnote{The definition of submodularity can be found in Definition \ref{def:Submodularity}.} in this context means that adding an additional CAV to a group of vehicles with significantly overlapping regions provides diminishing marginal gains in covered area. This property naturally discourages the system from including excessive overlaps in the set of actively transmitting vehicles. Specifically, we exploit this submodular property by formulating the sub-utility function for perception quality as $\mathcal{U}_{\mathrm{r}} = \sum_i \sum_{j = 1, j \ne i} \mathcal{P}_{ji} s_{ji} d_{ji}$, which reflects the accuracy and robustness of surrounding perception. Concurrently, the sub-utility function for the perceptual regions of CAVs is defined as $\mathcal{U}_{\mathrm{p}} = \sum_i \mathcal{I}(\mathcal{S}_i) $, which represents the amount of information collected by data fusion. }To maximize both utilities at the same time, the weighted utility function is formulated as:
\begin{equation}\label{eq:utility}
  \begin{aligned}
    \mathcal{U} _{\mathrm{sum}}(\mathcal{S},\mathcal{D}) &=\mathrm{\omega}_1\mathcal{U} _{\mathrm{r}} +\mathrm{\omega}_2\mathcal{U} _{\mathrm{p}} \\
    &=\mathrm{\omega}_1\underset{\text{Perception quality}}{\underbrace{\sum_{i=1}^N{\sum_{j=1,j\ne i}^N{\mathcal{P} _{ji}s_{ji}d_{ji}}}}}+\mathrm{\omega}_2\underset{\text{Perceptual region}}{\underbrace{\sum_{i=1}^N{\mathcal{I} \left( \mathcal{S} _i \right)}}}
,
  \end{aligned}
\end{equation}
where $\omega_1$ and $\omega_2$ denote the weights of perception quality and perceptual region, respectively. The optimized variable $\mathcal{D}=\left[d_{ij}\right]_{N\times N}$ is the matrix of data transmission rate. By combining the constraints and objective function Eq. (\ref{eq:utility}), we formulate the utility function maximization problem as:
\begin{equation}\label{OP:P0}
  \begin{aligned}
    \mathbf{P}: \quad \!&\max  \limits_{\mathcal{\mathcal{R}} ,\mathcal{S},\mathcal{D}} \quad \mathcal{U} _{\mathrm{sum}}(\mathcal{S},\mathcal{D})\\
    \textrm{s.t.} \quad
    &(\ref{eq:subchannel}), (\ref{eq:transmission rate}), (\ref{eq:the constraint of the compression ratio}), (\ref{eq:the constraint of the compression ratio 2}), (\ref{eq:computing power constraint}), (\ref{eq:power constraint}),
  \end{aligned}
\end{equation}
where the constraint (\ref{eq:subchannel}) is the upper bound of the number of subchannels; (\ref{eq:transmission rate}) denotes the constraint of transmission rate; (\ref{eq:the constraint of the compression ratio}) and (\ref{eq:the constraint of the compression ratio 2}) are the constraints of the compression ratio; (\ref{eq:computing power constraint}) and (\ref{eq:power constraint}) are the upper bounds of the computing capacity and transmission power, respectively. In Sec. \ref{sec:A Two-Stage Optimization Framework for Utility Maximization}, we prove that the utility function maximization can be decomposed into two problems, i.e., a nonlinear programming problem and submodular optimization problem that can be solved by a greedy algorithm with an approximation guarantee.

\section{Two-Stage Optimization Framework for Utility Maximization}
\label{sec:A Two-Stage Optimization Framework for Utility Maximization}
In this section, we first want to show that the original problem \(\mathbf{P}\) is NP-hard, making it hard to find its optimal solution within latency-sensitive systems. To find a suboptimal solution, we conceive a two-stage optimization framework, which decomposes \(\mathbf{P}\) into two distinct problems: a nonlinear programming (NLP) \(\mathbf{P_1}\), and a submodular optimization \(\mathbf{P_2}\). We ascertain the optimal solution for \(\mathbf{P_1}\) in Sec. \ref{sec:Nonlinear programming problem analysis}, and deduce a $(1 - e^{-1})$-approximation of the optimal value for \(\mathbf{P_2}\) using an iterative algorithm in Sec. \ref{sec:Submodular Analysis and Solutions}.

\subsection{Nonlinear Programming Problem Analysis}
\label{sec:Nonlinear programming problem analysis}
\begin{proposition}
\label{proposition:np-hard}
For fixed $\mathcal{D}^{(n)}$, problem \( \mathbf{P} \) is reducible from Weighted Maximum Coverage Problem (WMCP). When optimizing over matrices \(\mathcal{D}\) and \(\mathcal{R}\), \( \mathbf{P} \) surpasses WMCP's complexity, establishing its NP-hardness.
\end{proposition}

\begin{Proof}
Please refer to Appendix A.
{\hfill $\blacksquare$\par}
\end{Proof}
According to Proposition \ref{proposition:np-hard}, it can be concluded that the problem $\mathbf{P}$ is NP-hard. Besides, it can be observed that as the number of cooperative vehicles $N$ increases, the state space of problem $\mathbf{P}$ exhibits double exponential growth. Specifically, with the state space represented as $2^{\left(N^2\right)}\times \mathcal{D}_n \times \mathcal{R}_n$, the dimensionality increases rapidly with respect to $N$.  For example, let $\mathcal{D}_n$ and $\mathcal{R}_n$ denote the discrete levels for continuous variables $\mathcal{D}$ and $\mathcal{R}$. We can consider a scenario with 10 CAVs (refer to Sec. \ref{sec:simulation_settings}). As for the upper bound of $\mathcal{D}$ is $40$ Mbps, we assume a step size of 1 Mbps with a total of $\mathcal{D}_n = 40$ levels. As for the compression ratio $\mathcal{R}$, there are also $\mathcal{R}_n=20$ levels with an increment of 0.05 per level. Therefore, we can obtain the state space of the problem $\mathbf{P}$ is approximately $1.014\times 10^{33}$.

It is highly complicated to find the optimal solution even when the V2V network scale is no more than $10$ vehicles. Since the control and decision of CAVs are latency-sensitive, we have to conceive a real-time optimization solver to address the issue of finding an optimal solution for the NP-hard problem $\mathbf{P}$. Therefore, we decompose $\mathbf{P}$ into two subproblems by fixing one of the optimization variables.



Denote the link establishment in the $n$th round as $\mathcal{S}^{(n-1)}$. We then focus on adjusting the matrices for the compression ratio $\mathcal{R}$ and the data rate $\mathcal{D}$. Problem $\mathbf{P}_1$ is expressed as:
\begin{equation}\label{OP:P1}
  \begin{aligned}
    \mathbf{P_1}: \quad \!&\max \limits_{\mathcal{R},\mathcal{D}}\ \  \mathcal{U} _{\mathrm{sum}}\left(\mathcal{S}^{(n-1)}, \mathcal{D}\right)\\
    \textrm{s.t.} \quad
    &(\ref{eq:transmission rate}), (\ref{eq:the constraint of the compression ratio}), (\ref{eq:the constraint of the compression ratio 2}),\\
    &(\ref{OP:P1}\textrm{a}): A_j+\sum_{i=1,i\ne j}^N{{r} _{ij}s_{ij}^{(n-1)} d_{ij}}\leqslant F_j/\beta,\\
    &(\ref{OP:P1}\textrm{b}): \sum_{i=1,i\ne j}^N \left(E_{j}^{t}\big|_{s_{ij}^{(n-1)}}+ E_{ij}^{c}\big|_{s_{ij}^{(n-1)}} \right) \leq E^{T}_j.
  \end{aligned}
\end{equation}
For \(j = 1, 2, \ldots, N\), the sub-problem \(\mathbf{P}_1\) is an NLP problem due to the nonlinear constraints given by (\ref{eq:transmission rate}), (\ref{OP:P1}\textrm{a}), and (\ref{OP:P1}\textrm{b})\footnote{This is due to the product of decision variables \( {r} _{ij} \) and \( d_{ij} \).}. While global optimization techniques like branch and bound or genetic algorithms can be used for non-convex problems, they are computationally demanding. Our approach is to linearize the problem. We define \(\mathbf{U} = \mathcal{R} \odot \mathcal{D} = \left[ u_{ij} \right] _{N\times N}\) where \( \odot \) is the Hadamard product and \( u_{ij} = {r} _{ij}d_{ij} \). This linearizes the product term in the constraints. {\color{black}Thus, \(\mathbf{P}_1\) is equivalently reformulated as\footnote{For simplicity, we omit $\mathcal{I}\left(\mathcal{S}_i\right)$, $\omega_1$, and $\omega_2$ in the formulation of the first sub-problem, since those terms are independent of $\mathbf{P_{1-1}}$.}:
\begin{equation}\label{OP:P1-2}
  \begin{aligned}
    \mathbf{P}_{1-1}:\quad &\max_{\mathbf{U},\mathcal{D}} \,\,\sum_{j=1}^N{\sum_{i=1,i\ne j}^N{s_{ij}^{(n-1)}\mathcal{P} _{ij}d_{ij}}}\\
    \text{s.t.}\quad
    &(\ref{OP:P1-2}\text{a}): u_{ij} \le \min(C_{ij},{u_{ij}}{d_{ij}^{-1}}A_i), \\
    &(\ref{OP:P1-2}\text{b}): \max({r}_{j,\min},{\eta}{e^{-L_{ij}}}) \le {u_{ij}}{d_{ij}^{-1}} \le {r}_{j,\max},\\ 
    &(\ref{OP:P1-2}\text{c}): \sum_{i=1,i\ne j}^N{s_{ij}^{(n-1)}u_{ij}} \leqslant \min(\gamma_j ^{(n-1)},\varphi_j),
  \end{aligned}
\end{equation}
where we define \( \gamma_j ^{\left( n-1 \right)} \) as \( \frac{E^T}{\epsilon _j\tau _{j}^{c}} - \frac{\tau _{j}^{t}P_t\sum_{i=1,i\ne j}^N{s_{ij}^{(n-1)}}}{\epsilon _j\tau _{j}^{c}} - A_j \) and \( \varphi_j \) as \( \frac{F_j}{\beta} - A_j \). Moreover, the constraint (\ref{OP:P1-2}\textrm{c}) is derived from both (\ref{OP:P1}\textrm{a}) and (\ref{OP:P1}\textrm{b}). Even though we add a bilinear equality constraint with \( u_{ij} \), constraint (\ref{OP:P1-2}\textrm{b}) remains nonlinear, thereby we cannot obtain the optimal result directly. Given that \( \mathbf{P_{1-1}} \) attempts to optimize \( \sum_{j=1}^N{\sum_{i=1,i\ne j}^N{s_{ij}^{(n-1)}d_{ij}}} \), we have:
\begin{equation}\label{OP:P1-2-inequality}
  \begin{aligned}
    \frac{u_{ij}}{r_{j,\max}} \le d_{ij} \le \min \left\{ A_i, u_{ij} \left[ \max \left( r_{j,\min}, \frac{\eta}{e^{L_{ij}}} \right) \right]^{-1} \right\},
  \end{aligned}
\end{equation}
which offers an upper bound for the optimal value of \( \mathbf{P_{1-1}} \).} {\color{black}To optimize \( d_{ij} \), it is beneficial to focus on maximizing this limit. Therefore, we derive a relaxed problem as follows:
\begin{equation}\label{OP:P1-3}
  \begin{aligned}
    \mathbf{P_{1-2}}: \quad \!&\max \limits_{\mathbf{U}}\ \ \sum_{j=1}^N\sum_{i=1,i\ne j}^N  \frac{s_{ij}^{(n-1)}\mathcal{P} _{ij}u_{ij}}{\max \left( {r} _{j,\min},\frac{\eta}{e^{L_{ij}}} \right)}\\
    \textrm{s.t.} \quad
    &(\ref{OP:P1-2}\textrm{a}) \ \mathrm{and}\ (\ref{OP:P1-2}\textrm{c}),\\
    &u_{ij} = 0 \ \text{if} \ s_{ij}^{(n-1)} = 0.
  \end{aligned}
\end{equation}
\( \mathbf{P_{1-2}} \) is a standard linear programming solvable using techniques like the simplex or interior-point methods. If the optimal outcome of \( \mathbf{P_{1-2}} \) is \( u_{ij}^{(n)} \), the optimal solutions for transmission rate and adaptive compression ratio are \( d_{ij}^{(n)} = \min \left\{ A_i, u_{ij} \left[ \max \left( r_{j,\min}, \frac{\eta}{e^{L_{ij}}} \right) \right]^{-1} \right\}\) and \( {r}_{ij}^{(n)} = {u_{ij}^{(n)}}/{d_{ij}^{(n)}} \), respectively. Considering that \( d_{ij} \) values are at the edges of the feasible region, the optimal solution for \( \mathbf{P_{1-2}} \) matches that of \( \mathbf{P_{1-1}} \). It is noted that \( u_{ij} = 0 \ \text{if} \ s_{ij}^{(n-1)} = 0 \) guarantees that \( u_{ij} \) is explicitly set to zero whenever \( s_{ij}^{(n-1)} = 0 \).}
\subsection{Preliminaries for Submodular Optimization}\label{sec:Preliminaries}

Prior to delving into specific details of the other subproblem \( \mathbf{P_{2}} \), we briefly review the definition and primary characteristics of submodularity as presented in \cite{dobzinski2013communication}. 
\begin{definition}
\label{def:Derivative}
{\bfseries (Set Function Derivative)} Given a set function \( f: 2^{V} \rightarrow \mathbb{R} \), for a subset \( S \) of \( V \) and an element \( e \) in \( V \), the discrete derivative of \( f \) at \( S \) concerning \( e \) is denoted by \( \Delta_{f}(e | S) \) and defined as \( \Delta_{f}(e | S) = f(S \cup \{e\}) - f(S) \). If the context makes the function \( f \) evident, we omit the subscript, expressing it simply as \( \Delta(e | S) \).
\end{definition}

\begin{definition}\label{def:Monotonicity}
{\bfseries (Monotonicity)}
Given a function \(f: 2^{V} \rightarrow \mathbb{R}\), \(f\) is deemed monotone if, for all \(A, B \subseteq V\) with \(A \subseteq B\), the condition \(f(A) \leq f(B)\) holds.
\end{definition}
It should be underscored that the function \(f\) exhibits monotonicity if and only if every discrete derivative maintains a non-negative value. Specifically, for each \(A \subseteq V\) and any \(e \in V\), the relation \(\Delta(e | A) \geq 0\) is satisfied.

\begin{definition}
\label{def:Submodularity}
{\bfseries (Submodularity)} Let \( E \) denote a finite ground set. A set function \( f: 2^{E} \rightarrow \mathbb{R} \) is said to be normalized, non-decreasing, and submodular if it satisfies the following properties:
    \begin{enumerate}
        \item \( f(\emptyset) = 0 \);
        \item \(f\) is monotone as per Definition \ref{def:Monotonicity};
        \item  For any \(A, B \subseteq E\), \( f(A) + f(B) \geq f(A \cup B) + f(A \cap B) \);
        \item  For any \( A \subseteq B \subseteq E \) and an element \( e \in E \setminus B \), $\Delta _f(e|A)\geqslant \Delta _f(e|B)$.
    \end{enumerate}
\end{definition}

In the next subsection, we prove that the objective function possesses submodular properties. As more CAVs share their perception results, the ego CAV tends to add significant new information for view fusion. However, as the number of CAV increases, each additional CAV provides less new information, i.e., \textit{Diminishing Marginal Utility}. This concept is crucial in CAVs' scenarios, ensuring that resources are not wasted by redundant sensors.

\subsection{Submodular Analysis and Solutions}
\label{sec:Submodular Analysis and Solutions}
\begin{algorithm}[t]
    \caption{Greedy Algorithm for submodular function maximization}
    \label{alg:greedy-p2}
    \begin{algorithmic}[1]
      \REQUIRE Adaptive compression ratio matrix \( \mathcal{R}^{(n)} \), data transmission rate \( \mathcal{D}^{(n)} \).
      \ENSURE Output the optimal link establishment matrix \( \mathcal{S} \).
      
      \STATE Initialization: \( \mathcal{S} \leftarrow \emptyset \), \( i \leftarrow 1 \);
      \WHILE{\( i \leq N \)}
          \STATE \( s_{ij}^{*} = \arg \max_{s_{ij} \in \mathcal{S}^{n \times n} \backslash \mathcal{S}} \mathcal{U} _{\mathrm{sum}}(\mathcal{S} \cup \{s_{ij} \},\mathcal{D}^{(n)}) \);
          \STATE \( \mathcal{S} \leftarrow \mathcal{S} \cup \{s_{ij}^{*}\} \);
          \STATE \( i \leftarrow i+1 \);
      \ENDWHILE
      \RETURN \( \mathcal{S} \).
    \end{algorithmic}
\end{algorithm}

{\color{black}In this subsection, we first prove that the objective function of the problem $\mathbf{P_2}$ is a submodular function. The objective function $\mathcal{U}_{\mathrm{sum}}(\mathcal{S}, \mathcal{D})$ represents the aggregated utility of the system, incorporating both perception quality and coverage. However, achieving an optimal solution for this objective function is challenging due to its NP-hard nature. As more CAVs share their perception results, the ego CAV gains significant new information for view fusion. However, as the number of CAVs increases, each additional CAV provides less new information, i.e., \textit{Diminishing Marginal Utility}. This concept is crucial in CAV scenarios to ensure resources are not wasted by redundant sensors. To efficiently solve the submodular function maximization problem, we adapt a greedy algorithm. Submodular functions exhibit the diminishing returns property, allowing a greedy algorithm to find a near-optimal solution. Specifically, we show that a greedy algorithm can achieve at least $\left(1 - e^{-1}\right)$ of the optimal value to solve $\mathbf{P_2}$ efficiently.}
\begin{proposition}
\label{proposition:submodular}
Given that $\mathbf{P}_2$ characterizes the link establishment problem and the data rate $\mathcal{D}=\mathcal{D}^{(n)}$ is a constant matrix, the objective function $\mathcal{U}_{\text{sum}}(\mathcal{S},\mathcal{D}^{(n)})$ defined in Eq. (\ref{eq:utility}) exhibits submodularity if $\mathcal{U}_{\text{sum}}(\mathcal{S},\mathcal {D}^{(n)})$ satisfies all the properties outlined in Definition \ref {def:Submodularity} within Sec. \ref{sec:Preliminaries}.
\end{proposition}
\begin{proof}
    Please refer to Appendix \ref{proof: Proposition 2}.{\hfill $\blacksquare$\par}
\end{proof}

With the adaptive compression ratio matrix \( \mathcal{R}^{(n)}=\left[{r}_{ij}^{(n)}\right]_{N\times N} \) and data transmission rate \( \mathcal{D}^{(n)}=\left[d_{ij}^{(n)}\right]_{N\times N} \), we then formulate the link establishment matrix \( \mathcal{S} \) for \( \mathbf{P}_2 \) as:
\begin{equation}\label{OP:P2}
  \begin{aligned}
    \mathbf{P}_2:\quad \!&\max_{\mathcal{S}} \,\, \sum_{i=1}^N{\left( \mathrm{\omega}_1\sum_{j=1,j\ne i}^N{\mathcal{P} _{ji}s_{ji}d_{ij}^{(n)}}+\mathrm{\omega}_2\mathcal{I} \left( \mathcal{S} _i \right) \right)}\\
    \textrm{s.t.} \quad
    &(\ref{OP:P2}\textrm{a}): \sum_{i=1,i\ne j}^N{\mathrm{\chi}_{ij}^{\left( n \right)} s_{ij}}\le E_{j}^{T}-\tau _{j}^{c}\epsilon _jA_j ,\\
    &(\ref{OP:P2}\textrm{b}): \sum_{i=1,i\ne j}^N{s_{ij} u_{ij}^{(n)}}\leqslant \varphi_j \ \mathrm{and}\  (\ref{eq:subchannel}), 
  \end{aligned}
\end{equation}
where $\mathrm{\chi}_{ij}^{\left( n \right)}=u_{ij}^{\left( n \right)}\epsilon _j\tau _{j}^{c}+\tau _{j}^{t}P_t$ can be obtained by the inequality constraint (\ref{eq:power constraint}). Since the objective function $\mathcal{U} _{\mathrm{sum}}(\mathcal{S},\mathcal{D})$ is a submodular function according to Proposition \ref{proposition:submodular}, $\mathbf{P_{2}}$ is a submodular function maximization problem, which can be solved by a greedy algorithm for near-optimal results. 

\begin{proposition}
\label{proposition:gap}
Given a submodular, non-decreasing set function $\mathcal{U} _{\mathrm{sum}}(\mathcal{S},\mathcal{D}^{(n)})$, which yields \(\mathcal{U} _{\mathrm{sum}}(\emptyset,\mathcal{D}^{(n)})=0\), the greedy algorithm obtains a set \(\mathcal{S}_{G}\) satisfying:
\begin{equation}\label{eq:gap}
  \begin{aligned}
    \mathcal{U} _{\mathrm{sum}}\left( \mathcal{S} _G,\mathcal{D} ^{(n)} \right) \ge \left( 1-e^{-1} \right) \max_{\mathcal{S}} \mathcal{U} _{\mathrm{sum}}\left( \mathcal{S} ,\mathcal{D} ^{(n)} \right). 
  \end{aligned}
\end{equation}
\end{proposition}
\begin{Proof}
    Please refer to Sec. II of \cite{krause2014submodular}.{\hfill $\blacksquare$\par}
\end{Proof}\par%
According to Proposition \ref{proposition:gap}, it can be observed that Algorithm \ref{alg:greedy-p2} can obtain a $(1-e^{-1})$-approximation of the optimal value of $\mathbf{P_{2}}$.
During the \(n\)th round, we update the link establishment \( \mathcal{S}^{(n)} \). If link \( (i,j) \) reduces the throughput, it is removed: \( \mathcal{S}^{(n)}\gets \mathcal{S}^{(n)}\backslash\{s_{ij}\} \). Otherwise, it is added: \( \mathcal{S}^{(n)}\gets \mathcal{S}^{(n)}\cup{\{s_{ij}\}} \). For \( \mathbf{P_{2}} \), we iteratively adjust links until we find an optimal solution meeting all constraints or hit the iteration limit by relying on Algorithm \ref{alg:greedy-p2}.

This greedy approach not only ensures maximization of the objective function \( \mathcal{U}_{\mathrm{sum}}(\mathcal{S},\mathcal{D}^{(n)}) \) but also significantly reduces computational complexity, making it highly advantageous for real-time applications. According to the cardinality constraint $K$ in (\ref{eq:subchannel}), the time-complexity of Algorithm \ref{alg:greedy-p2} is only $\mathcal{O}(K)$.  
Therefore, we can circumvent the complexities typically associated with optimization problems by such a greedy algorithm rather than engaging in exhaustive searches or iterative procedures, which may not always guarantee convergence to the optimal solution.

\subsection{Deep Learning-Based Adaptive Compression}
\label{sec:Deep Learning-Based Adaptive Compression}
\begin{figure}[t]
  \centering
  \includegraphics[width=0.47\textwidth]{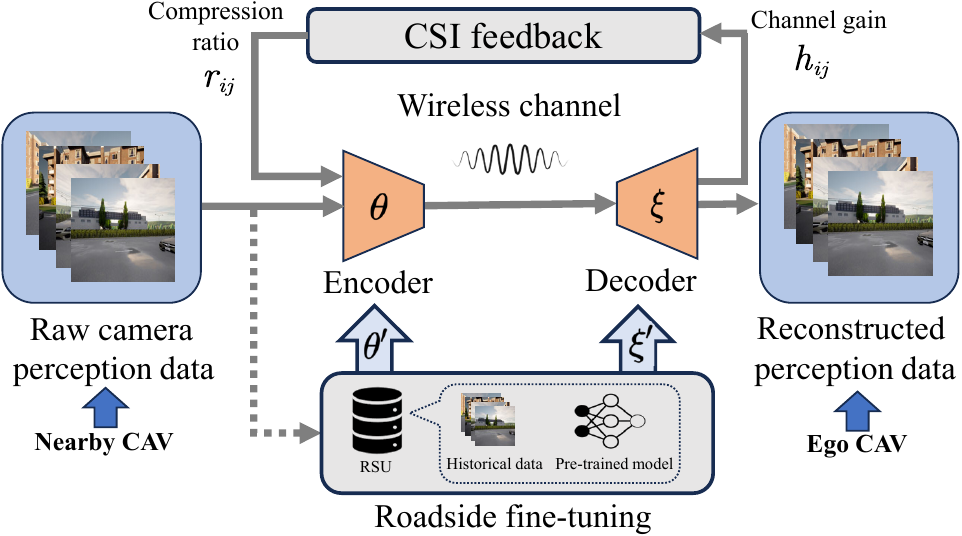}
  \caption{{Overall architecture of adaptive compression.} }
  \label{fig: adaptive compression}
\end{figure}

The compression modules commonly employed in V2V collaborative networks often rely on predefined fixed compression ratios, such as JPEG and JPEG2000\cite{xu2022cobevt}. However, these fixed ratios are inadequate to accommodate the demands of the dynamic channel conditions discussed in Sec. \ref{sec:A Two-Stage Optimization Framework for Utility Maximization}. Moreover, it has been shown that the deep learning-based method generally offers better rate–distortion (R-D) performance compared to the standard compression methods\cite{balle2016end}. In this section, we propose an adaptive compression method that comprises an adaptive R-D mechanism to refine the compression ratio, $\mathcal{R}$, to align with the requirements of dynamic channel conditions. Then, we introduce a fine-tuning strategy to reduce temporal redundancy in V2V transmissions by exploiting the holistic RGB frames. We first present the main procedure of our DL-based compression scheme as follows.

The Deep Learning-Based Compression (DBC) addresses the challenges by using trainable parameters from training datasets, thus offering adaptability to V2V dynamics. Under DBC architecture, both encoder and decoder utilize convolutional layers. The encoder transforms the input image $\mathbf{x}$ into a latent representation, \( \mathbf{z} = f(\mathbf{x};{\theta}) \), with transformation parameters \( \theta \) learned from training. The decoder uses a distinct parameter set \( \xi \) for reconstruction: \( \tilde{\mathbf{z}} = H(\mathbf{z};{\xi}) \). $\mathbf{\hat{x}}$ is the reconstruction image from the decoder. The training objective is to minimize:
\begin{equation}\label{eq:basic R-D equation}
  \mathop{\arg\min}_{\theta, \xi} R\left(\mathbf{\tilde{z}}; \theta\right)+\beta D(\mathbf{x}, \mathbf{\hat{x}}; \theta, \xi),
\end{equation}
where \( R\left(\mathbf{\tilde{z}}; \theta\right) = \mathbb{E} \left[ -\log_2 p_\mathbf{\tilde{z}} \left(\mathbf{\tilde{z}}\right) \right] \) represents the rate function and \( D\left(\mathbf{x}, \mathbf{\hat{x}}; \theta, \xi\right) = \mathbb{E} \left[ \left\|\mathbf{x} - \mathbf{\hat{x}}\right\|^2 \right] \) denotes the distortion function,where the constant \( \beta \) manages the R-D trade-off.  Moreover, we formulate the traditional fixed R-D problem as a multi-R-D problem for adaptability:
\begin{equation}\label{eq:trained R-D}
  \mathop{\arg\min}_{\theta,\xi}R(\hat{h};\theta )+\beta D\left(\mathbf{x}, \mathbf{\hat{x}}; \theta, \xi \right).
\end{equation}
It is noted that \( \beta \) can affect reliable decoding and image quality, emphasizing the need to adjust \( \beta \) adaptively. In this context, we introduce a DBC mechanism to dynamically modify \( \beta \) under dynamic channel conditions (As shown in Fig. \ref{fig: adaptive compression}, given the channel state information (CSI) feedback, we can obtain \( \beta \) and compression ratio according to Sec. \ref{sec:Nonlinear programming problem analysis}). Therefore, the revised problem can be formulated as:
\begin{equation}\label{eq:optimized multi R-D}
  \mathop{\arg\min}_{\theta,\xi} R(\hat{h};\theta,{r} )+G({r}) D\left(\mathbf{x}, \mathbf{\hat{x}}; \theta, \xi,{r} \right),  
\end{equation}
where function \( G \) is the R-D mapping function, which utilizes a lookup table and converts the compression ratio to the tradeoff parameter \( \beta \). Therefore, there is no explicit expression for the function \( G \). Moreover, the adaptive compression network is based on a pre-trained model\cite{yang2020variable}, using historical camera data as training dataset to obtain network parameter $\theta,\xi$ by solving the multi-R-D problem in (\ref{eq:optimized multi R-D}). The detailed architecture and processes are described in Appendix \ref{appendix: Network Structure Details}. We also evaluate the R-D performance of our proposed PACP framework using the Multi-scale Structural Similarity (MS-SSIM) and Peak Signal-to-Noise Ratio (PSNR) metrics in Appendix \ref{appendix: The Evaluation of Rate-Distoration (R-D) Performance}.

As for computational cost, the encoder and decoder of MAE require 0.155 MFLOPs/pixel and 0.241 MFLOPs/pixel\footnote{As a metric of computational cost, MFLOPs/pixel denotes the number of million floating point operations performed per pixel for camera perception data.}, respectively\cite{yang2020variable}. Considering Tesla FSD as the computational unit, when three CAVs share their perception results at a rate of 10 fps, the average latencies for encoding and decoding processes are observed to be 40.26 ms and 20.89 ms, respectively.
\begin{algorithm}[t]
    \caption{Priority Aware Collaborative Perception}
    \label{alg:pacs-throughput}
    \begin{algorithmic}[1]
      \REQUIRE Multi-CAV data \( \varGamma \), number of vehicles \( N \), channel parameter constraints \( K, C_{ij} \), device parameters \( {r}_{j,\min}, {r}_{j,\max}, \eta, \tau_j^t, \tau_j^c, E_j^T \), etc.
      \ENSURE Priority weight \( \mathcal{P} \), near-optimal compression ratio \( \mathcal{R} \), link establishment \( \mathcal{S} \), data rate \( \mathcal{D} \), modulated autoencoder, and BEV prediction.
      
      \STATE Initialize priority weight \( \mathcal{P}^{(0)} \) by equally allocating bandwidth and transmitting initial frames;
      \STATE Initialize the link establishment decision \( \mathcal{S}^{(0)} \);
      
      \FOR{\( j = 0 \) to \( N-1 \)}
        \STATE Sort column of link establishment decision matrix in descending order and get the indices of the largest $K$ capacity, and set the associated $s_{ij}=1$;
      \ENDFOR
      
      \WHILE {Convergence not achieved}
        \STATE Solve the linear programming problem \( \mathbf{P_{1-2}} \);
        \STATE Solve the submodular problem \( \mathbf{P_{2}} \) using Algorithm \ref{alg:greedy-p2};
        
        \STATE Update priority weight \( \mathcal{P} \) based on the current bandwidth allocation and perception results;
        
        \IF {\color{black}{Priority weight \( \mathcal{P} \) changes exceed the threshold $\epsilon _{\mathcal{P}}$}}
          \STATE {\color{black}Re-initialize the link establishment decision \( \mathcal{S}^{(0)} \), then re-calculate initial priority weight \( \mathcal{P} \);}
          \STATE {\color{black}Continue the optimization iteration for problem \( \mathbf{P} \);}
        \ENDIF
        
      \ENDWHILE
      
      \STATE Obtain the near-optimal solution for \( \mathcal{R} ,\mathcal{S},\mathcal{D} \);
      \STATE Determine the trade-off parameter \( \beta^{*} \) based on the channel state;
      \STATE Train encoders and decoders to compress and reconstruct raw camera data according to Eq.~(\ref{eq:optimized multi R-D});
      \STATE Predict BEV feature using the reconstructed camera data. 
    \end{algorithmic}
\end{algorithm}

\section{Simulation Results and Discussions}
\label{sec:Simulation Results And Discussions}
In this section, we evaluate our schemes under various communication settings, which consist of bandwidth, transmission power, the number of CAVs, and the distribution of vehicles. Subsequently, we delve into the comparison of raw data reconstruction performance—both with and without the application of a fine-tuned compression strategy. In the final part of this section, we present the results of BEV along with the associated IoU.

\subsection{Dataset and Baselines}
\label{subsec:Dataset and Baselines}
\textbf{Dataset:} To validate our approach, we employ the CARLA and OpenCOOD simulation platforms, exploiting the OPV2V dataset\cite{xu2022opv2v}. This dataset encompasses 73 varied scenes, a multitude of CAVs, 11,464 frames, and in excess of 232,000 annotated 3D vehicle bounding boxes. All of these have been amassed using the CARLA simulator\cite{dosovitskiy2017carla}.

\textbf{Baseline 1:} the Fairness Transmission Scheme (\textbf{FTS}), which is built on the principles laid out in \cite{9681261}. The cornerstone of this approach is the allocation of subchannels in a manner that resonates with Jain index defined in Eq. (\ref{Jain's Fairness Index}).

\textbf{Baseline 2:} The core of this baseline is the Distributed Multicast Data Dissemination Algorithm (\textbf{DMDDA}), as outlined in \cite{lyu2022distributed}, which seeks to optimize throughput in a decentralized fashion. 

\textbf{Baseline 3:} the \textbf{No Fusion} scheme, which implies the usage of a single ego vehicle for gathering information about its surroundings. It operates without integrating data from the cameras of proximate CAVs.

To ensure a fair comparison, we maintain uniformity in the transmission model and simulation parameters, aligning them to those discussed in Sec. \ref{sec:Channel Modeling}.

\subsection{Simulation Settings}
\label{sec:simulation_settings}
Our simulations are based on the 3GPP standard\cite{8891313}. Specifically, the communication range for vehicles is 150 m. Vehicular speeds, unless otherwise stated, ranging from $0$ to $50$ km/h (typical speed range on city roads), are generated using the CARLA simulator\cite{dosovitskiy2017carla}. 
Additionally, vehicles are uniformly distributed over a six-lane highway of 200 meters, with three lanes allocated for each traveling direction. To simplify, certain frame types, such as cyclic redundancy checks and Reed-Solomon coding, have been omitted from our simulation. Table \ref{table:simulation_parameters} presents the main simulation parameters utilized in our research.
\begin{table}[t]
\centering
\caption{Simulation Parameters}
\begin{tabular}{|c|c|}
\hline
\textbf{Parameter} & \textbf{Value} \\
\hline
\hline
Number of vehicles (\(N\)) & 10 \\
Local data per vehicle (\(A_j\)) & 40 Mbits \\
Number of subchannels (\(K\)) & 4 \\
Computation complexity (\(\beta\)) & 100 Cycles/bit \\
Transmission power (\(P_t\)) & 8 mW \\
CPU capacity (\(F_j\)) & 1 GHz - 3 GHz \\
Bandwidth (\(W\)) & 200 MHz \\
Power threshold (\(E_j^T\)) & 1 kW \\
Weight for perception quality (\(w_1\)) & \(1 \times 10^{-2}\) \\
Weight for perceptual region (\(w_2\)) & \(1 \times 10^{-3}\) \\
Parameter (\(\eta\)) & 1 \\
Compression ratio range (\(r_{\min}, r_{\max}\)) & (0.3, 0.95) \\
\hline
\end{tabular}
\label{table:simulation_parameters}
\end{table}

\subsection{Analysis of Experimental Results}
\label{sec:Experimental Results}
\begin{figure}[t]
  \centering
  \subfigure[IoU vs. Noise level]{
  \includegraphics[width=4.1 cm]{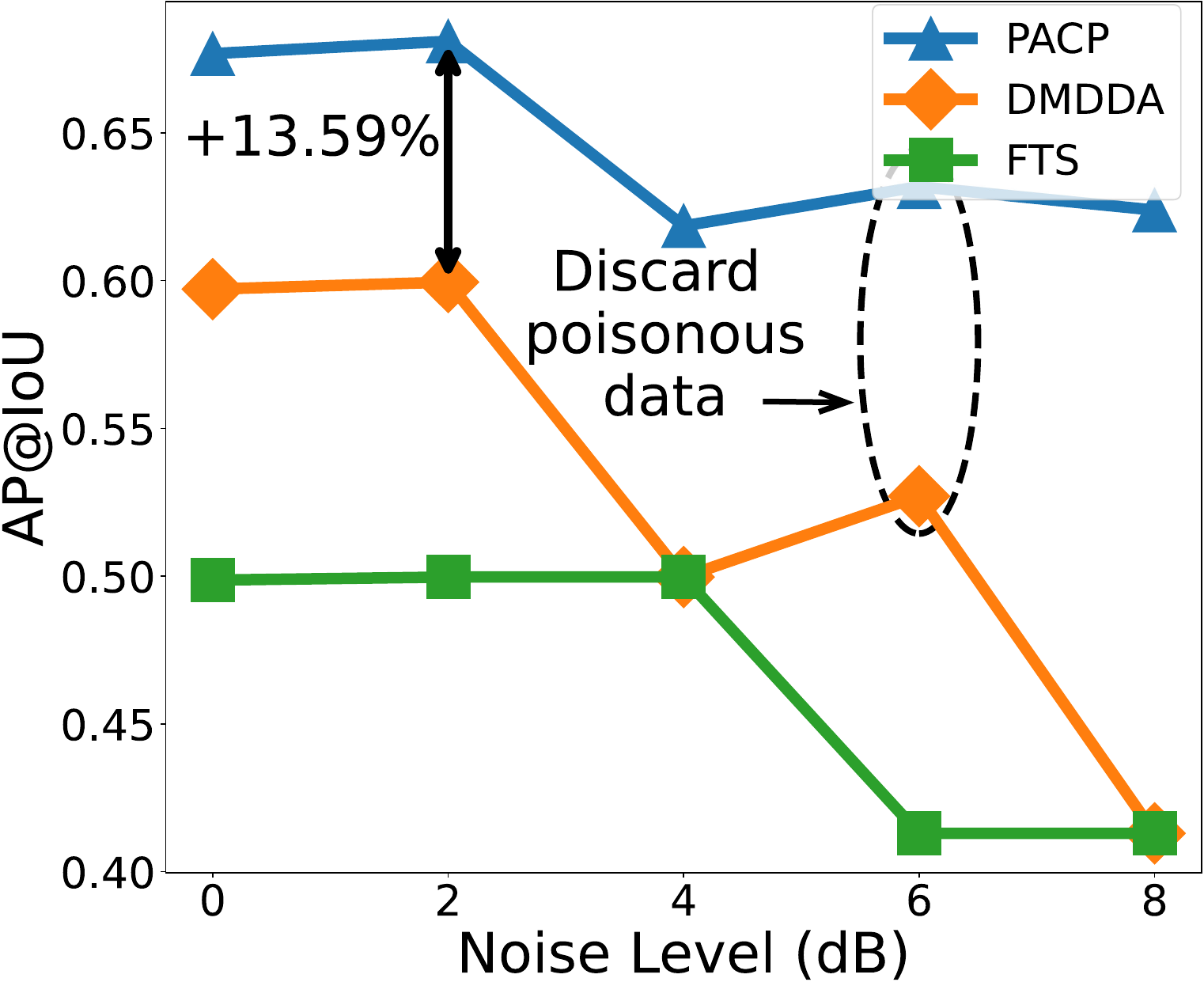}\label{fig:noise-a}
  }
  \subfigure[Utility vs. Noise level]{
  \includegraphics[width=4.1 cm]{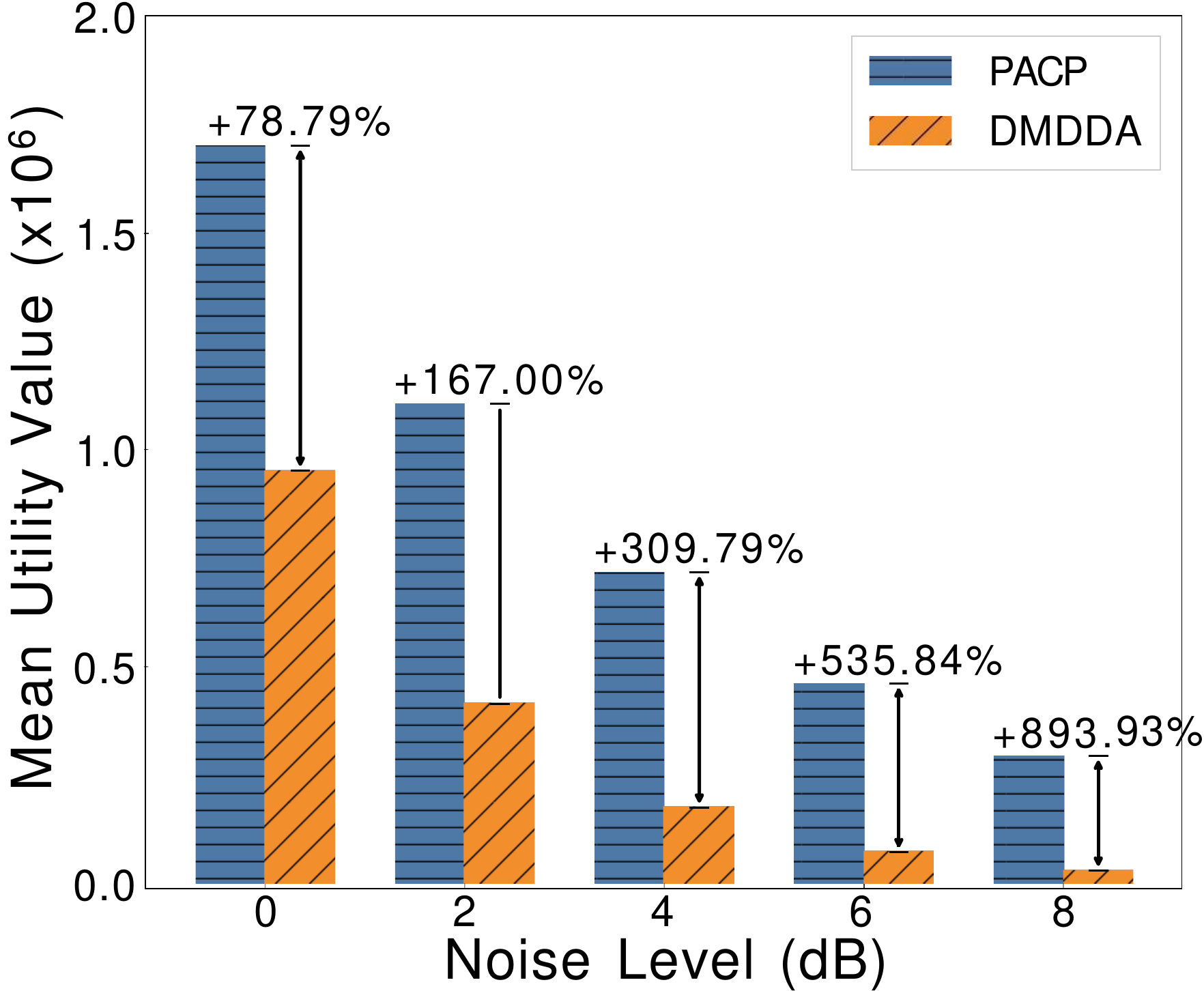}\label{fig:noise-b}
  }
  \caption{AP@IoU and utility value under different noise levels.}
  \label{fig:noise}
\end{figure}
\begin{figure*}[t]
  \centering
  \subfigure[$\mathcal{U} _{\mathrm{sum}}$ vs. Power]{
  \includegraphics[width=3.9cm]{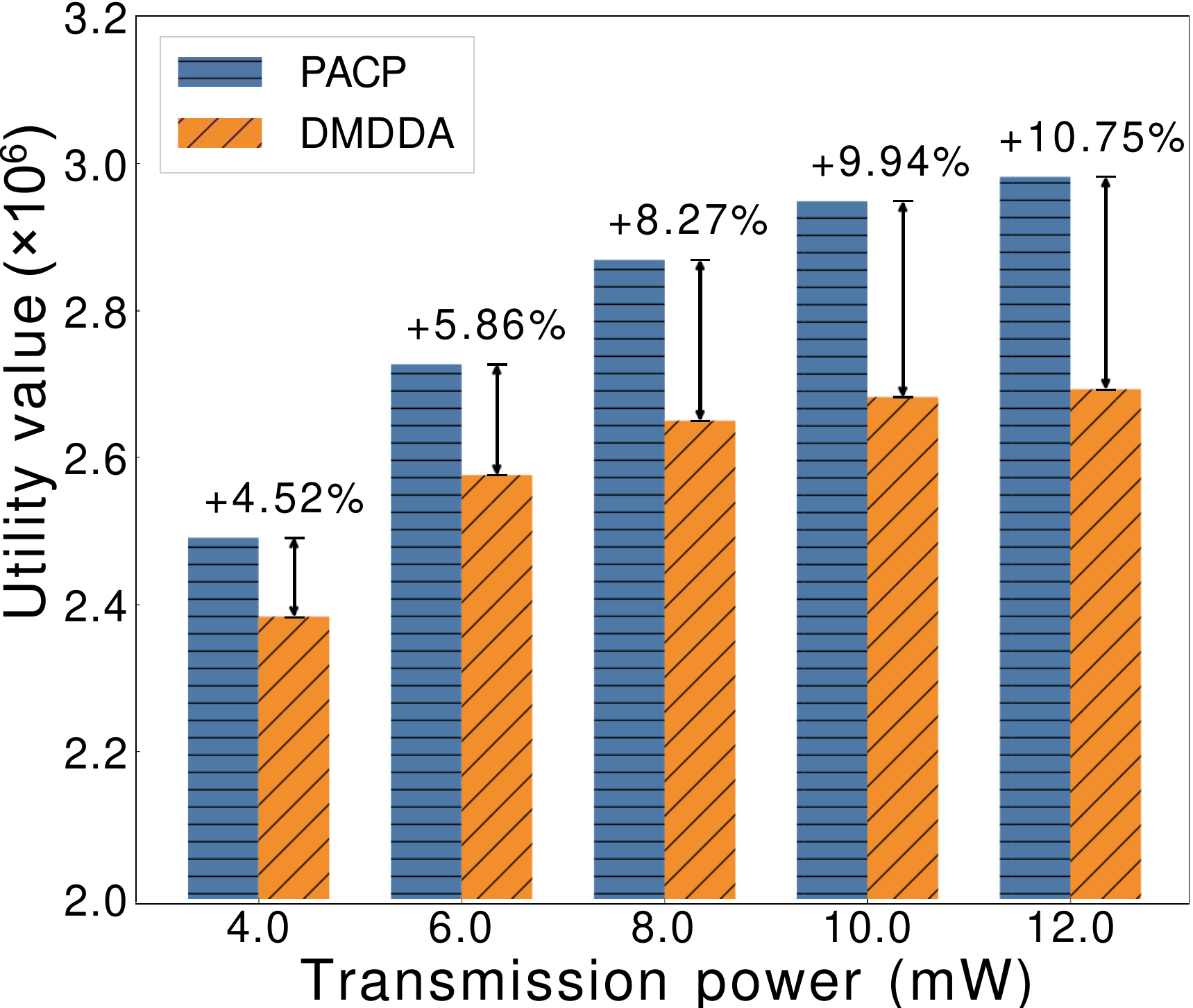}\label{fig:utility_power}
  }
  \subfigure[$\mathcal{U} _{\mathrm{sum}}$ vs. $R_{\max}$]{
  \includegraphics[width=3.9cm]{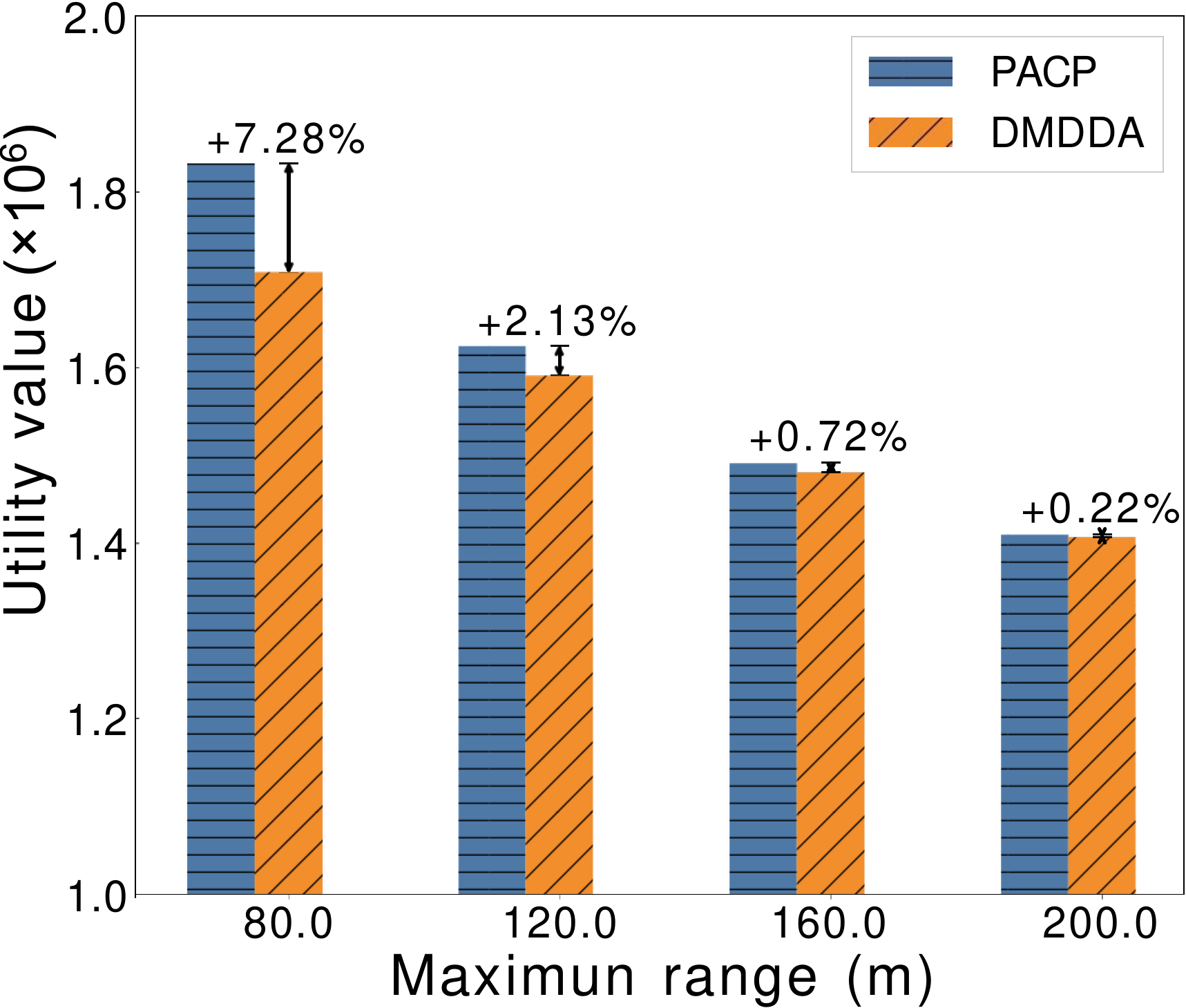}\label{fig:utility_range}
  }
  \subfigure[$\mathcal{U} _{\mathrm{sum}}$ vs. CAV num.]{
  \includegraphics[width=3.9cm]{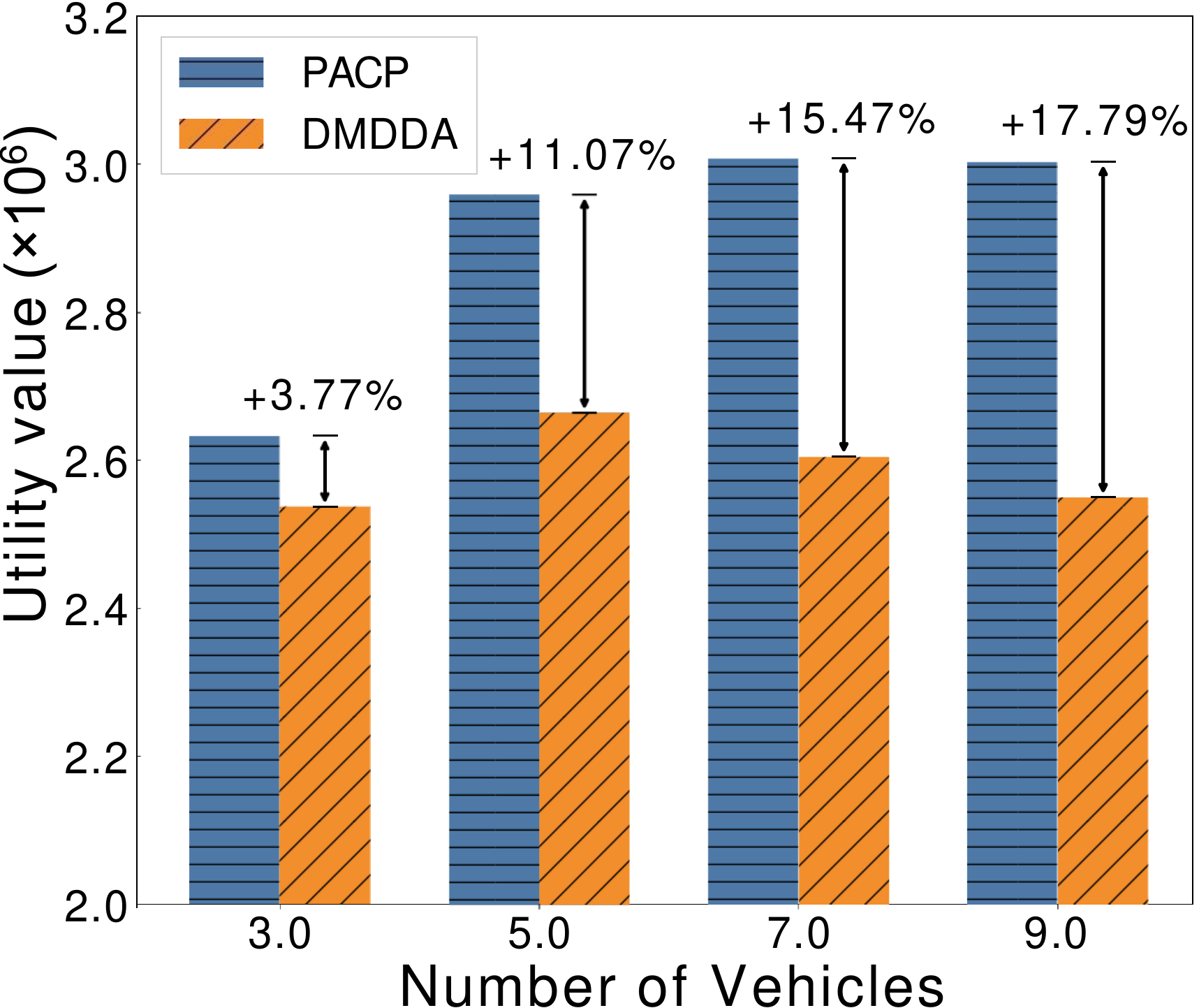}\label{fig:utility_num}
  }
  \subfigure[$\mathcal{U} _{\mathrm{sum}}$ vs. Bandwidth]{
  \includegraphics[width=3.9cm]{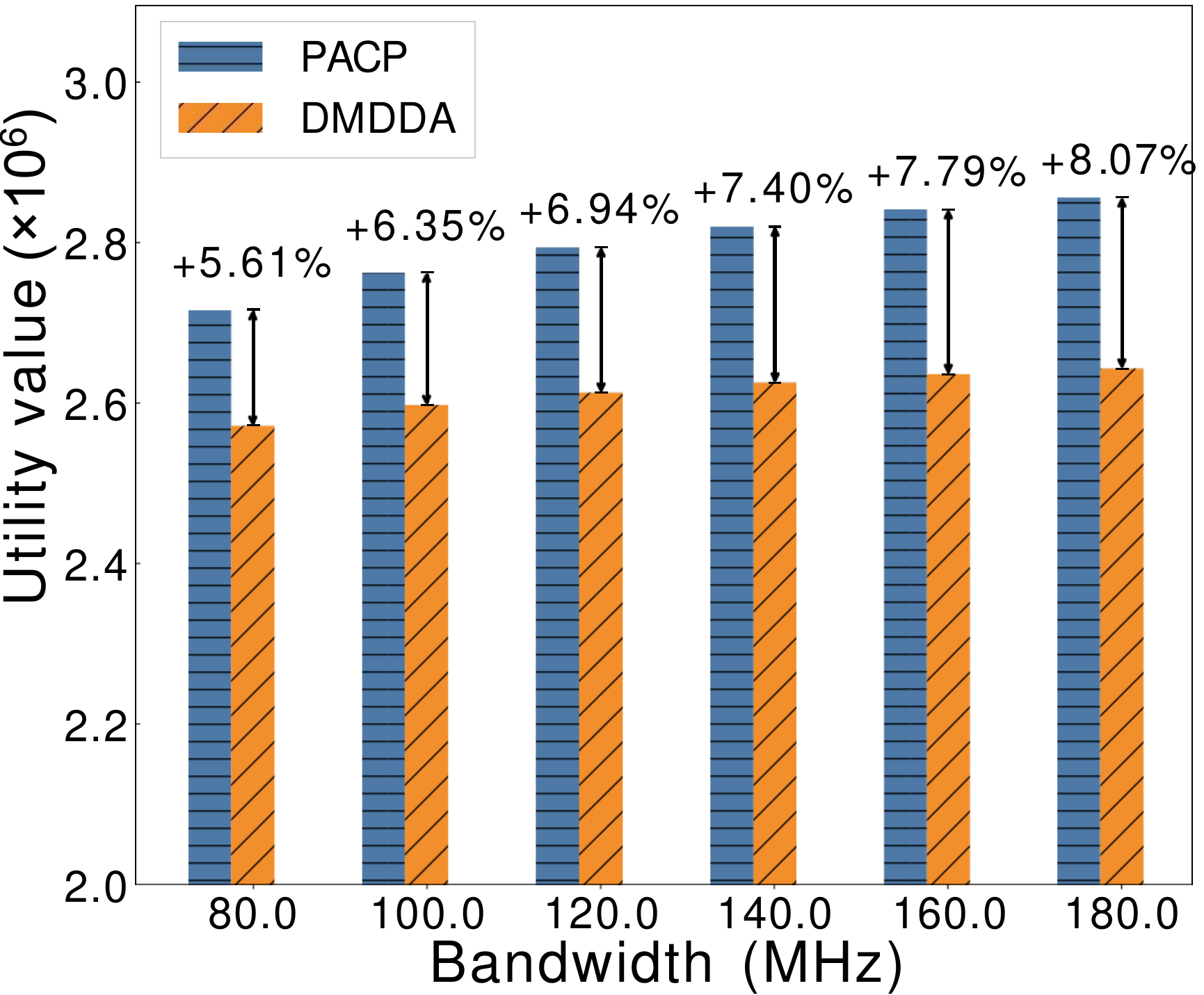}\label{fig:utility_bandwidth}
  }
  \caption{Results on average utility value $\mathcal{U} _{\mathrm{sum}}$ under different communication settings.}
  \label{fig:utility_diff}
\end{figure*}
\begin{figure*}[t]
  \centering
  \subfigure[Throughput vs. Power]{
  \includegraphics[width=3.9cm]{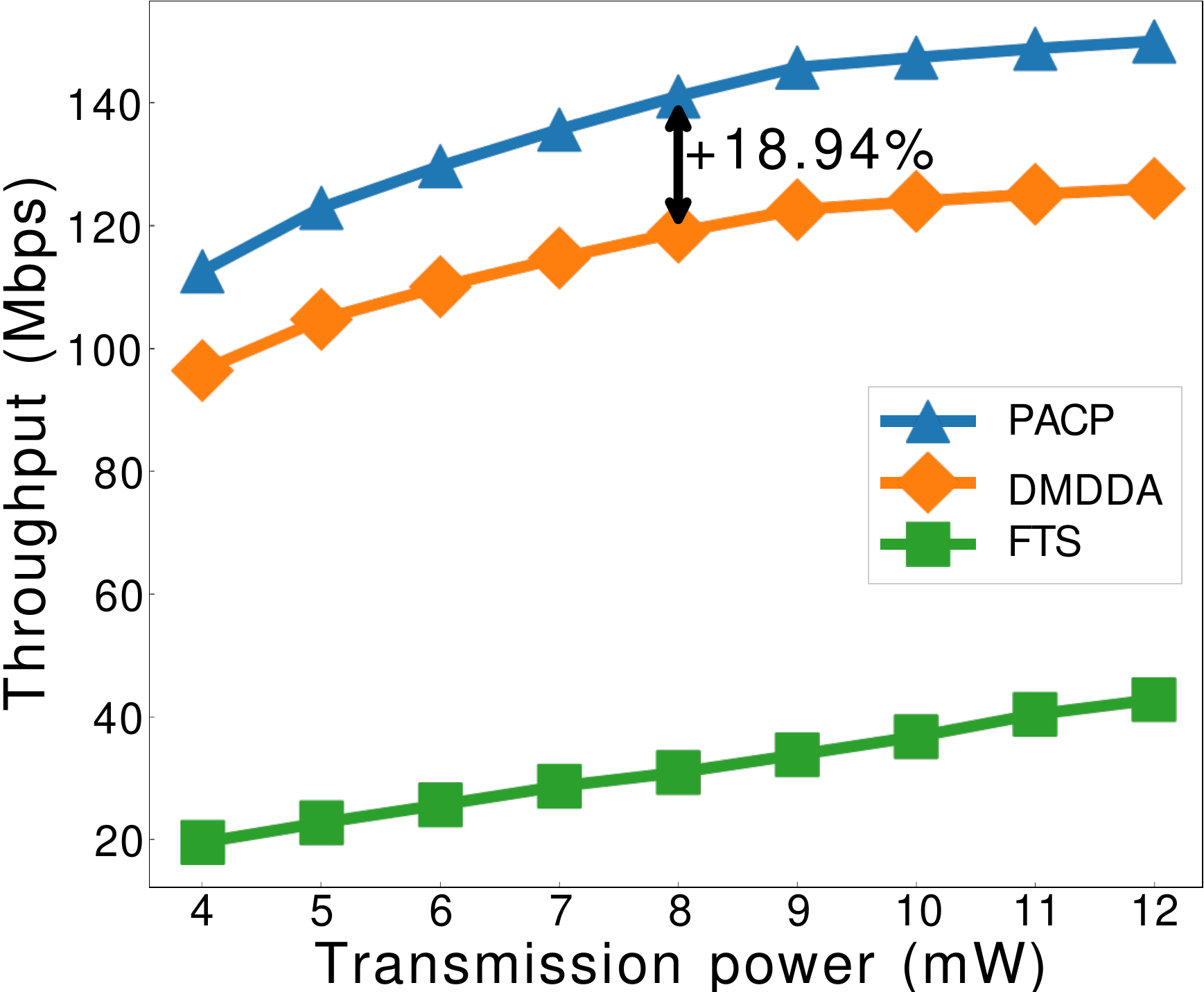}\label{fig:throughput_power}
  }
  \subfigure[Throughput vs. $R_{\max}$]{
  \includegraphics[width=3.9cm]{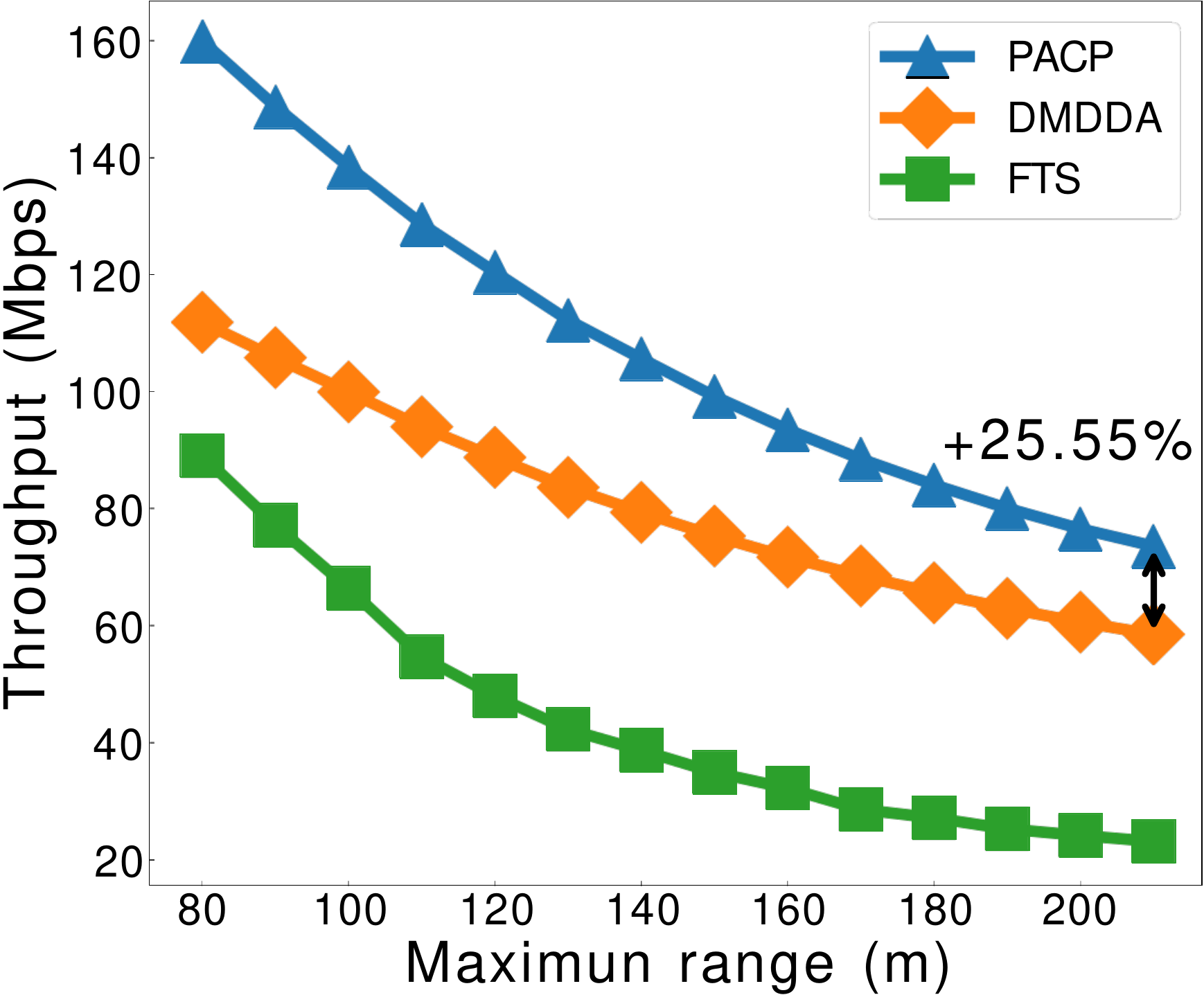}\label{fig:throughput_range}
  }
  \subfigure[Throughput vs. CAV num.]{
  \includegraphics[width=3.9cm]{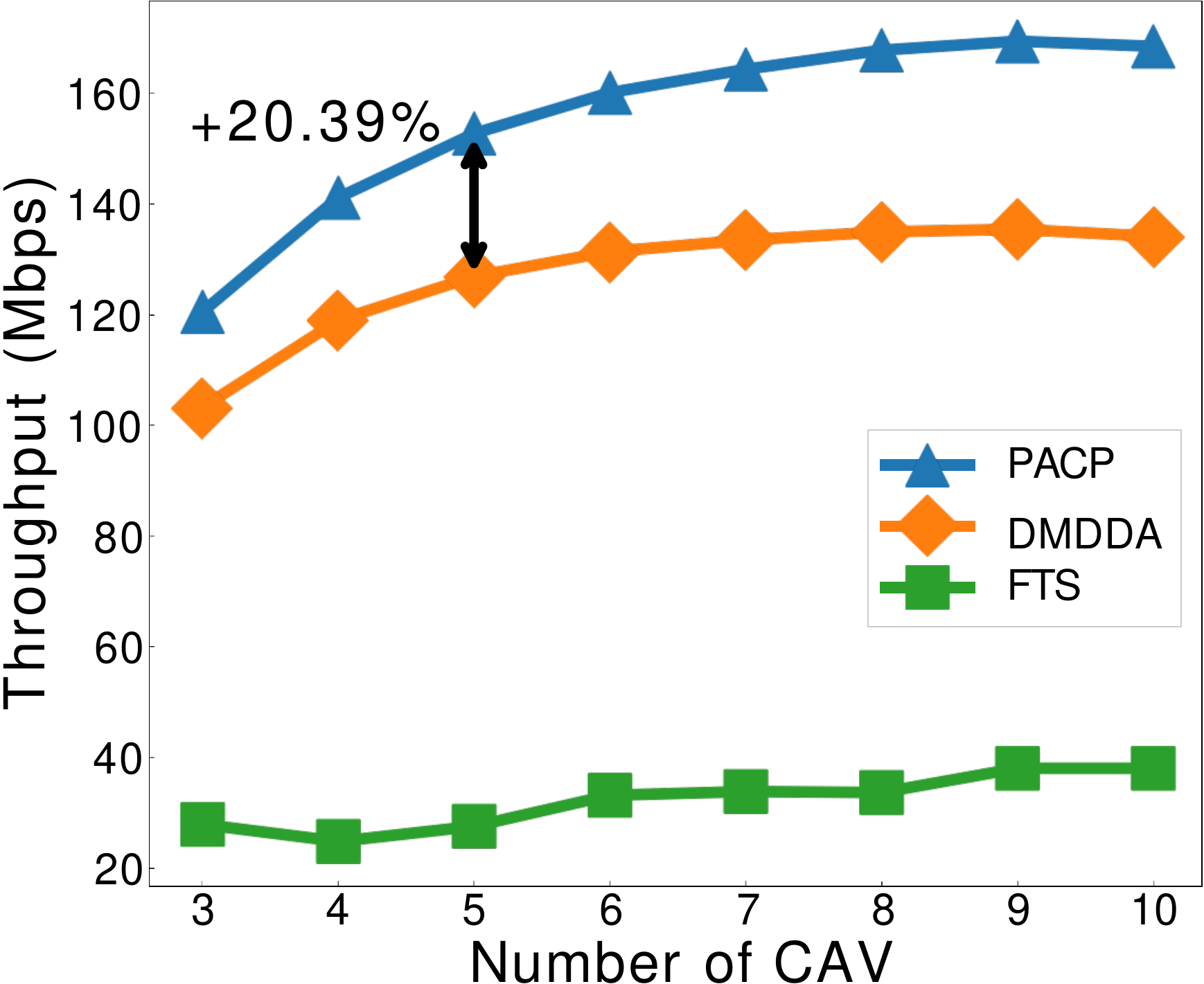}\label{fig:throughput_num}
  }
  \subfigure[Throughput vs. Bandwidth]{
  \includegraphics[width=3.9cm]{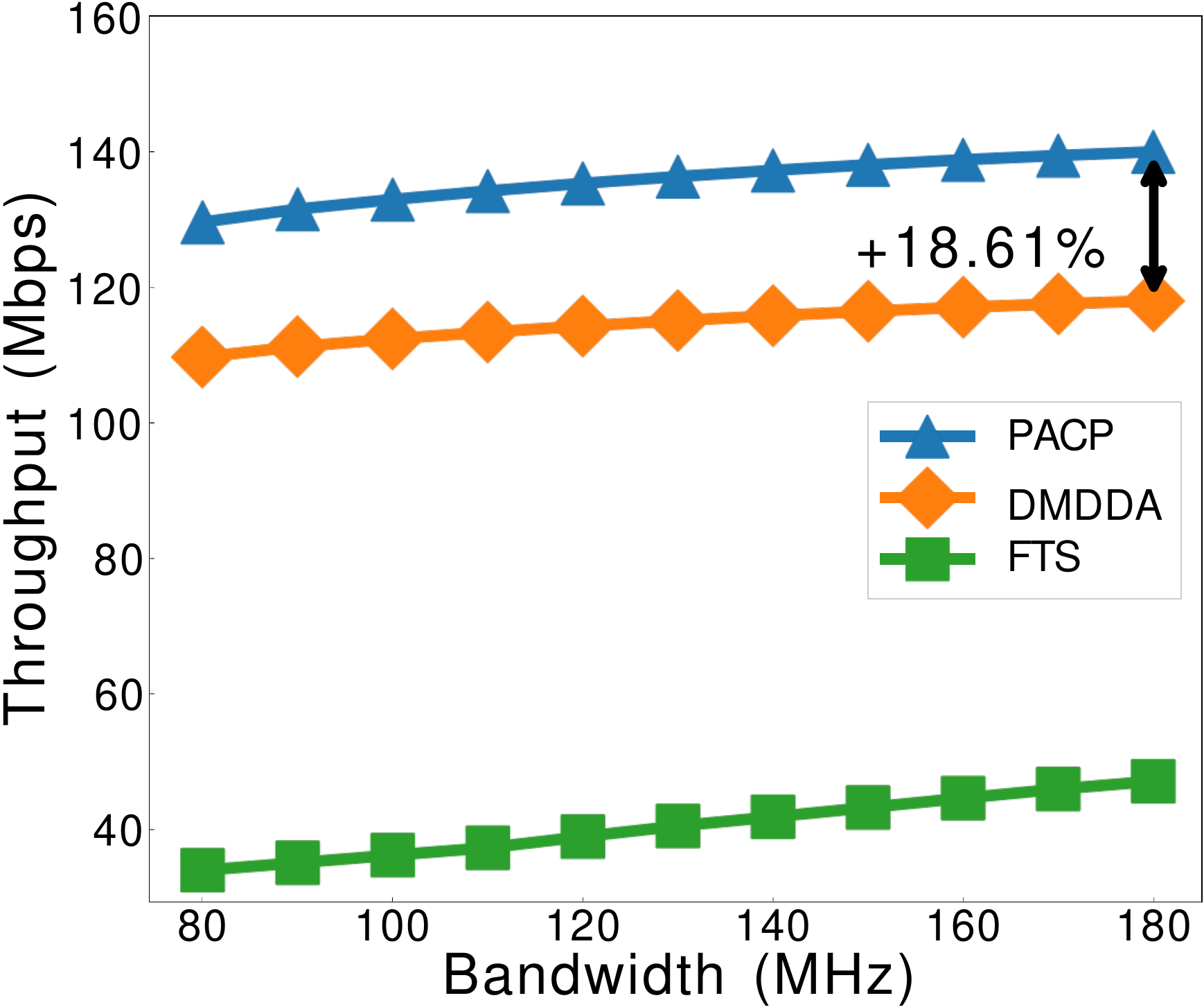}\label{fig:throughput_bandwidth}
  }
  \caption{Results on average network throughput under different communication settings.}
  \label{fig:throughput_comm}
\end{figure*}
In our study, we evaluate the efficacy of the proposed PACP algorithm in maximizing utility value and enhancing perception accuracy (AP@IoU). We benchmark its performance against two established methods: DMDDA\cite{lyu2022distributed} and FTS\cite{9681261}. For fairness, all evaluations are carried out under consistent conditions. Firstly, we conducted experiments to investigate the impact of electromagnetic interference on collaborative perception. The interference levels were classified based on their relative noise floor according to urban environment: 0 dB over -174 dBm/Hz (Low), 4 dB over -174 dBm/Hz (Medium), and 8 dB over -174 dBm/Hz (High). As illustrated in Fig. \ref{fig:noise}, perception accuracy and utility value deteriorate sharply as noise level intensifies. While PACP outperforms both DMDDA and FTS in terms of AP@IoU. Specifically, PACP shows an improvement of 14.04\% over DMDDA and 36.58\% over FTS, because PACP utilizes the adaptive compression to reduce packet loss and latency. In Fig. \ref{fig:noise-a}, due to the priority-aware scheme, PACP effectively discards poisonous data, thereby mitigating interference adverse impact on perception. In Fig. \ref{fig:noise-b}, PACP demonstrates superior performance compared to DMDDA in terms of utility value by increasing over 78.79\%. 

In Eq. (\ref{eq:utility}), the utility value $\mathcal{U} _{\mathrm{sum}}$ denotes the weighted utility function of perception quality and perceptual region. As the priority value is determined by the degree of IoU matching, a larger utility value indicates that more data matching the perception results of the ego CAV have been transmitted. Intuitively, the higher the utility value, the higher the accuracy of the perception results of the ego CAV. Fig. \ref{fig:utility_diff} presents a comprehensive view of how different communication parameters influence the utility value $\mathcal{U}_{\mathrm{sum}}$. In Fig. \ref{fig:utility_power}, as the transmission power increases, both PACP and DMDDA show an improvement in the utility value. However, the utility value of PACP consistently surpasses that of DMDDA. Specifically, at a typical power 8.0 mW, PACP is about 8.27\% better than DMDDA, underlining the efficiency of the proposed PACP approach over the baseline in terms of better prioritized data communication. Fig. \ref{fig:utility_range} indicates that the utility value generally declines as the maximum range increases. Despite this trend, PACP consistently outperforms DMDDA, illustrating its robustness even when the communication range is extended. Fig. \ref{fig:utility_num} reveals that as the number of CAVs increases, the utility value increases. After deploying 9 vehicles, the PACP algorithm outstrips DMDDA by a notable margin of 17.79\%. This accentuates the scalability and efficacy of PACP in larger vehicular networks. From Fig. \ref{fig:utility_bandwidth}, PACP amplifies the utility value by 8.07\% in comparison with DMDDA at a bandwidth of 180 MHz. This suggests that PACP is adept at capitalizing on larger bandwidths to enhance the perception accuracy of the ego CAV. {\color{black}In Fig. \ref{fig:utility_energy}, as the energy consumption constraints \(E^T\) increase, both PACP and DMDDA show an improvement in the utility value. However, the utility value of PACP consistently surpasses that of DMDDA. Specifically, at a typical constraint of 10 W, PACP achieves about 21.34\% higher utility value than DMDDA.}

Intuitively, we can adopt throughput as a metric for cooperative perception. Similar to utility value, a higher throughput often entails a greater amount of valuable information being transmitted. Fig. \ref{fig:throughput_power} shows that larger power enhances throughput. Fig. \ref{fig:throughput_range} depicts an inverse correlation between the vehicle distribution maximum range (80 m - 210 m) and throughput. It is noted that the performance of FTS strategy deteriorates because of balancing communication channels of remote vehicles. Remarkably, with five cooperative CAVs, the PACP algorithm augments the average throughput by 20.39\% compared with DMDDA. Furthermore, Fig. \ref{fig:throughput_bandwidth} underscores the linear affinity between the network's bandwidth (80-180 MHz) and throughput. Specifically, a bandwidth of $W=180$ MHz elevates the average throughput by 18.61\% using the PACP algorithm against other baselines. To sum up, it is evident that a higher throughput results in a richer amount of high-priority data being transmitted according to Figs. \ref{fig:utility_diff}-\ref{fig:throughput_comm}. 

\begin{table*}[]
\centering
\caption{The comparison of PACP with three baselines under different parameters in terms of perception accuracy.}
\label{tab:table1}
\resizebox{0.9\textwidth}{!}{
\begin{tabular}{|c|ccc|ccc|ccc|}
\hline
\multirow{2}{*}{\diagbox[width=8em]{AP@IoU}{Parameters}} & \multicolumn{3}{c|}{\textbf{Power (mW)}} & \multicolumn{3}{c|}{\textbf{Num. of CAVs}} & \multicolumn{3}{c|}{\textbf{Bandwidth (MHz)}} \\ \cline{2-10}
& \textbf{5} & \textbf{8} & \textbf{11} & \textbf{2} & \textbf{3} & \textbf{4} & \textbf{80} & \textbf{140} & \textbf{200} \\ \hline
\textbf{No Fusion} & 0.408 & 0.408 & 0.408 & 0.408 & 0.408 & 0.408 & 0.408 & 0.408 & 0.408 \\ \hline
\textbf{FTS} & 0.410 & 0.499 & 0.499 & 0.506 & 0.447 & 0.410 & 0.410 & 0.410 & 0.499 \\ \hline
\textbf{DMDDA} & 0.590 & 0.605 & 0.646 & 0.648 & 0.605 & 0.603 & 0.597 & 0.597 & 0.603 \\ \hline
\textbf{PACP} & \textbf{0.662} & \textbf{0.673} & \textbf{0.684} & \textbf{0.685} & \textbf{0.685} & \textbf{0.685} & \textbf{0.672} & \textbf{0.680} & \textbf{0.685} \\ \hline
\textbf{Improvement} & 12.20\% & 11.24\% & 5.88\% & 5.71\% & 13.22\% & 13.60\% & 12.56\% & 13.90\% & 13.60\% \\ \hline
\end{tabular}
}
\end{table*}
\begin{figure*}[t]
  \centering
  \includegraphics[width=2\columnwidth]{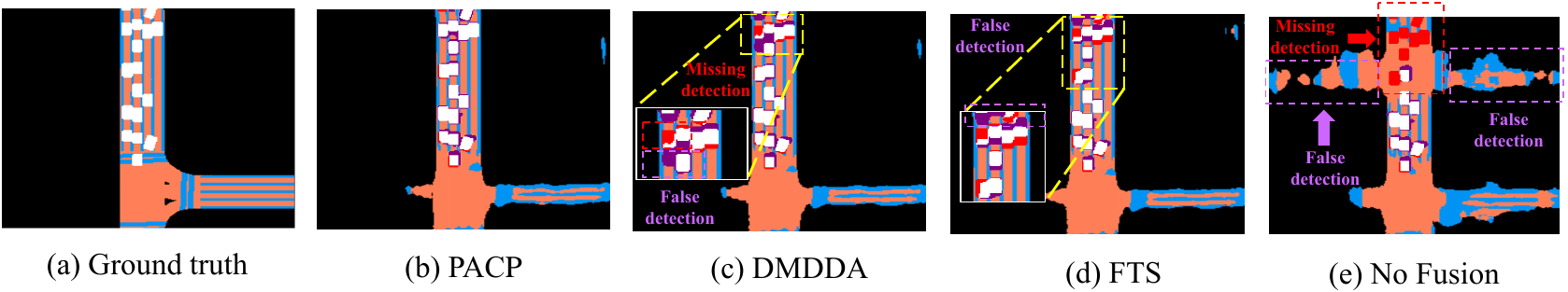}
  \caption{BEV prediction: (a) Ground Truth, (b) PACP, (c) DMDDA, (d) No Fusion. PACP accurately detects dynamic objects, unlike DMDDA, FTS, and No Fusion. White: correct detection; \textcolor{red}{\textbf{Red}}: missing detection; \textcolor{newextractedpurple}{\textbf{Purple}}: false detection.}
  \label{fig:bev}
\end{figure*}
\begin{figure}[t]
  \centering
  \includegraphics[width=0.47\textwidth]{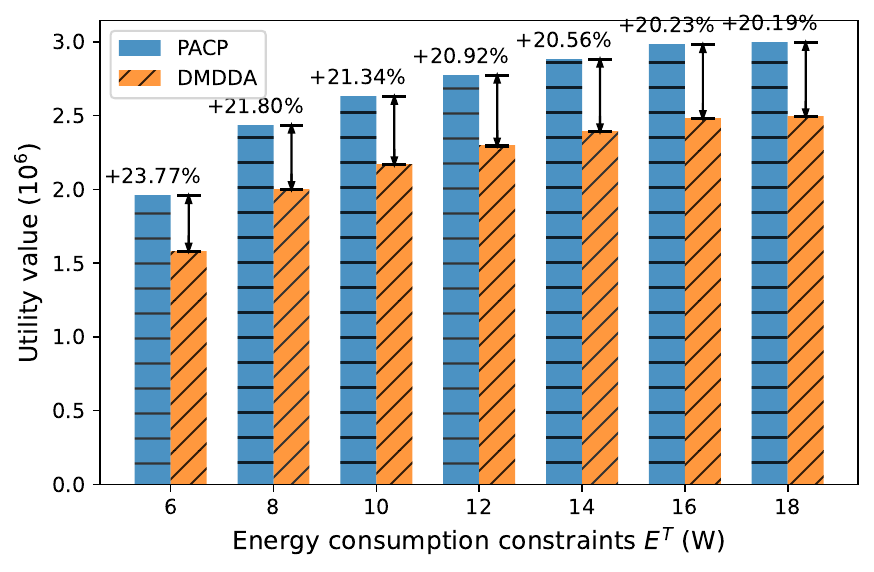}
  \caption{The Utility value vs Energy consumption constraints $E^T$ (W).}
  \label{fig:utility_energy}
\end{figure}
To further analyze the performance of perception accuracy, we conduct extensive experiments to verify the performance of PACP. From Table \ref{tab:table1}, it is evident that the proposed PACP method consistently outperforms other techniques under various settings. PACP, underpinned by a priority-aware perception scheme, demonstrates a remarkable increase in perception accuracy, surpassing FTS by a minimum of 35.38\% and DMDDA by 5.71\% in scenarios under varying numbers of CAVs. Given a typical bandwidth 200 MHz, the improvement of adapting PACP achieves up to 13.60\%. Such performance elevation is attributed to PACP's adeptness at channel allocation based on CAVs' significance to the ego vehicle, ensuring optimal transmission of pivotal perception data. Contrarily, fairness-based strategies, as exemplified by FTS's bandwidth-fairness and DMDDA's subchannel-fairness, inadvertently lead to suboptimal resource allocation by treating every CAV with uniform significance. This often results in the situation that vital information receives the same resource allocation as non-essential data, consequently diminishing AP@IoU and increasing the chances of missed detections in a BEV context. Not unexpectedly, as the power, the number of CAVs, and bandwidth increase, the amount of information perceived by the ego CAV expands, leading to a consequent rise in the AP@IoU for all perception strategies. 

Fig. \ref{fig:bev} illustrates the BEV predictions across different schemes, setting an emphasis on the deviations from the Ground truth (GT), depicted in Fig. \ref{fig:bev}(a), which symbolizes the paragon of collaborative perception communication. Observing Fig. \ref{fig:bev}(b), the PACP stands out by registering zero false or missing detections, adhering closely to GT. On the other hand, DMDDA and FTS schemes, represented in Figs. \ref{fig:bev}(c) and (d), manifest a notable number of missing detections, highlighting the inefficacies in their perception algorithms. The most striking deviations are observed in the No Fusion scheme, as portrayed in Fig. \ref{fig:bev}(e). This scheme, relying solely on single vehicle perception, results in an escalated number of both false and missing detections. Such stark disparities accentuate the imperativeness of leveraging multi-view data fusion for optimized perception outcomes and further illuminate the pronounced efficacy of the PACP scheme over its contemporaries.

\section{Conclusions}
\label{sec:Conclusion}
In this paper, we have investigated the perception challenges posed by the inherent limitations of the BEV, particularly its blind spots, in ensuring the safety and reliability of connected and autonomous vehicles. We have introduced the novel PACP framework that leverages a unique BEV-match mechanism, discerning priority levels by computing the correlation between the captured information from nearby CAVs and the ego vehicle. By employing a two-stage optimization based on submodular optimization, PACP optimally regulates transmission rates, link connectivity, and compression metrics. A distinguishing feature of PACP is its integration with a deep learning-based adaptive autoencoder, which is further enhanced by fine-tuning mechanisms, ensuring efficient image reconstruction quality under dynamic channel conditions. Evaluations conducted on the CARLA platform with the OPV2V dataset have confirmed the superiority of PACP to existing methodologies. The results have demonstrated that PACP improves the utility value and AP@IoU by \textbf{8.27\%} and \textbf{13.60\%}, respectively, highlighting its potential to set new benchmarks in the realm of collaborative perception for connected and autonomous vehicles.

\begin{appendices}
\section{The Background Information on IoU}
Intersection over Union (IoU) is an effective metric in computer vision, especially in object detection, used to quantify the overlap between two areas. Though IoU fundamentally informs the design of our priority weight system in subsequent sections, it is directly applied as a validation tool in our experimental setup to assess the accuracy of perception overlaps among Connected Autonomous Vehicles (CAVs).

\begin{figure}[h]
    \centering
    \subfigure[The intersection of $G$ and $P$.]{
        \includegraphics[width=4.1cm]{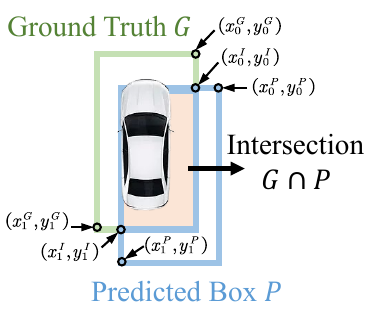}\label{fig:intersection}
    }
    \subfigure[The union of $G$ and $P$.]{
        \includegraphics[width=4.1cm]{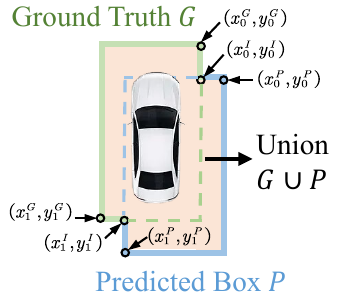}\label{fig:Union}
    }
    \caption{Illustration of Intersection over Union (IoU) using a ground truth bounding box $G$ and a predicted bounding box $P$.}
    \label{fig: iou}
\end{figure}

As shown in Fig. \ref{fig: iou}, IoU is calculated between a ground truth bounding box \(G\) and a predicted bounding box \(P\), which represent the actual and camera-based perceived locations, respectively. The coordinates for \(G\) (\((x_0^G, y_0^G)\) and \((x_1^G, y_1^G)\)) and \(P\) (\((x_0^P, y_0^P)\) and \((x_1^P, y_1^P)\)) denote their respective top-left and bottom-right corners.

The respective areas of \(G\) and \(P\) are defined as follows:
\begin{equation}
\begin{aligned}
S_G &= |(x_0^G - x_1^G) \times (y_0^G - y_1^G)|, \\
S_P &= |(x_0^P - x_1^P) \times (y_0^P - y_1^P)|.
\end{aligned}
\end{equation}

The area of overlap, critical for the IoU calculation, is derived from the extremities of the overlapping coordinates:
\begin{equation}
S_{G\cap P} = \left| (x_0^I - x_1^I) \times (y_0^I - y_1^I) \right|
\end{equation}

The union area, combining both bounding boxes minus their intersection, is given by:
\begin{equation}
S_{G \cup P} = S_G + S_P - S_{G \cap P}
\end{equation}

IoU is thus calculated as:
\begin{equation}
\mathrm{IoU}(G, P) = \frac{S_{G\cap P}}{S_{G\cup P}} = \frac{\left| (x_0^I - x_1^I) \times (y_0^I - y_1^I) \right|}{S_G + S_P - S_{G\cap P}}.
\end{equation}

The performance of IoU is heavily influenced by the quality of the communication channel. Superior channel quality enhances the transmission of camera data, leading to more precise and synchronized data fusion among CAVs. This synchronization reduces latency in data exchanges, directly improving the alignment of perceived and actual locations, thus boosting IoU performance. Conversely, poor channel conditions can decrease overlap accuracy between predicted and ground truth boxes, thereby reducing IoU scores. This emphasizes the critical role of robust channel quality in achieving optimal perception system performance in CAV networks.

\section{Proof of Proposition 1}\label{proof: Proposition 1}
\begin{Proof}
Define the Weighted Maximum Coverage Problem (WMCP) with universe \( U \), subset family \( \mathcal{F}=\left\{\hat{S}_{1}, \ldots, \hat{S}_{F}\right\} \), and weight function \( w: U \rightarrow \mathbb{R}_{\geq 0} \). For set \( \hat{S} \subseteq U \) and subset family \( \mathcal{Y} \subseteq \mathcal{F} \):
\[
w(\hat{S}) = \sum_{u \in \hat{S}} w(u),
\]
where $w(\mathcal{Y}) = w\left(\bigcup_{\hat{S} \in \mathcal{Y}} \hat{S}\right)$. WMCP aims to select subsets from \( \mathcal{F} \) to maximize the weight of their elements:
\begin{equation}\label{eq:WMCP}
    \max_{\hat{S}} \sum_{u \in \hat{S}} w(u) \quad \text{s.t.} \quad \left| \hat{S} \right|<K, \hat{S}_f\subseteq U.
\end{equation}
According to the definition of the utility function in Eq. (\ref{eq:utility}), we define the weight function with fixed $\mathcal{D} ^{(n)}$ as follows:
\begin{equation}\label{eq:utility_i}
    \mathcal{U} _i(\mathcal{S} _i,\mathcal{D} ^{(n)})=\mathrm{\omega}_1\sum_{j=1,j\ne i}^N{\mathcal{P} _{ji}s_{ji}d_{ij}^{(n)}}+\mathrm{\omega}_2\mathcal{I} \left( \mathcal{S} _i \right).
\end{equation}
It is noted that each element in \( U \) is a unique link for a CAV in \( \mathbf{P} \). The weight \( w \) in WMCP aligns with the utility \( \mathcal{U} _{i}(\mathcal{S}_i,\mathcal{D}^{(n)}) \) in \( \mathbf{P} \). Hence, with fixed $\mathcal{D} ^{(n)}$, the problem can be seen as WMCP. Considering the optimization over \(\mathcal{D}\) and \(\mathcal{R}\), \(\mathbf{P}\) extends beyond WMCP, introducing complexities from \(\mathcal{D}\) and \(\mathcal{R}\).
Given the inherent NP-hardness of WMCP and the added complexity when including optimization over \(\mathcal{D}\) and \(\mathcal{R}\), we can infer that problem \(\mathbf{P}\) is, at the very least, as hard as the classical WMCP. Hence, problem \(\mathbf{P}\) is NP-hard. {\hfill $\blacksquare$\par}
\end{Proof}

\section{Proof of Proposition 2}\label{proof: Proposition 2}
\begin{Proof}
Refer to four properties from Definition \ref{def:Submodularity} in Sec. \ref {sec:Preliminaries}, we will prove that the objective function of the problem $\mathbf{P}$ satisfies all essential conditions for submodular functions as follows:

\textbf{(1)} To show the first property of Definition \ref{def:Submodularity}, we can obtain $f\left(\emptyset\right)=\mathcal{U} _{\mathrm{sum}}(\mathcal{S},\mathcal{D})=0$ if $\mathcal{S}=\mathbf{0}$, which represents that we do not establish any links so $\mathcal{S}$ is a zero matrix. Therefore, the first property of Definition \ref{def:Submodularity} can be satisfied.

\textbf{(2)} To show the second property of Definition \ref{def:Submodularity}, we 
 assume $\mathbf{B}_1$ and $\mathbf{B}_2$ are link establishment sets, where $\forall \mathbf{B}_1\subseteq \mathbf{B}_2\subseteq \mathbf{E}$. Let $\left\{s_{ij}^{\mathbf{B}_1}\right\}_{N \times N}$ and $\left\{s_{ij}^{\mathbf{B}_2}\right\}_{N \times N}$ denote the associated link establishment matrices. Thus, we have:
\begin{equation}\label{eq:Monotonicity1}
\begin{aligned}
    f\left( \mathbf{B}_1 \right) =\mathcal{U} _{\mathrm{sum}}(\mathbf{B}_1,\mathcal{D} )=\sum_{i=1}^N{\left( \mathrm{\omega}_1\sum_{j=1,j\ne i}^N{\mathcal{P} _{ji}s_{ji}^{\mathbf{B}_1}d_{ij}^{(n)}}+\mathrm{\omega}_2\mathcal{I} \left( \mathbf{B}_1 \right) \right)}.
\end{aligned}
\end{equation}
Similarity, 
\begin{equation}\label{eq:Monotonicity2}
\begin{aligned}
    f\left( \mathbf{B}_2 \right) =\sum_{i=1}^N{\left( \mathrm{\omega}_1\sum_{j=1,j\ne i}^N{\mathcal{P} _{ji}s_{ji}^{\mathbf{B}_2}d_{ij}^{(n)}}+\mathrm{\omega}_2\mathcal{I} \left( \mathbf{B}_2 \right) \right)}.
\end{aligned}
\end{equation}
Then, let $\mathbf{B}_2\backslash\mathbf{B}_1$ represent the set of all elements that are in \( \mathbf{B}_2 \) but not in \( \mathbf{B}_1 \).
Therefore, we can derive the difference between $\mathbf{B}_1$ and $\mathbf{B}_2$ as $f\left( \mathbf{B}_2\backslash \mathbf{B}_1 \right) =\sum_{i=1}^N{\left( \mathrm{\omega}_1\sum_{j=1,j\ne i}^N{\mathcal{P} _{ji}s_{ji}^{\mathbf{B}_2\backslash \mathbf{B}_1}d_{ij}^{(n)}}+\mathrm{\omega}_2\mathcal{I} \left( \mathbf{B}_2\backslash \mathbf{B}_1 \right) \right)}$. 
Each element of the matrices \( s_{ij}^{\mathbf{B}} \) can only take values of 0 or 1, where 1 signifies a connected link and 0 denotes a broken or non-existent link. Given \( \forall \mathbf{B}_1\subseteq \mathbf{B}_2\subseteq \mathbf{E} \), it is evident that every link connected in \( \mathbf{B}_1 \) is also connected in \( \mathbf{B}_2 \). However, the converse may not be true; there might be links in \( \mathbf{B}_2 \) that are not connected in \( \mathbf{B}_1 \). When \( s_{ij}^{\mathbf{B}_2} = s_{ij}^{\mathbf{B}_1} \), the difference is zero. When \( s_{ij}^{\mathbf{B}_1} \neq s_{ij}^{\mathbf{B}_2} \), the difference is positive. Therefore, the difference $f\left( \mathbf{B}_2\backslash\mathbf{B}_1 \right) \geq 0$, and the objective function $\mathcal{U} _{\mathrm{sum}}(\mathcal{S},\mathcal{D})$ is monotone.

(3) To show the third property of Definition \ref{def:Submodularity}, we have to prove \( f(\mathbf{B}_1) + f(\mathbf{B}_2) \geq f(\mathbf{B}_1 \cup \mathbf{B}_2) + f(\mathbf{B}_1 \cap \mathbf{B}_2) \). According to Eqs. (\ref{eq:union}) and (\ref{eq:Monotonicity1}), we can decompose $f(\mathcal{S}_i)$ into two distinct components as follows:
\begin{equation}\label{eq:submodular4-32}
\begin{aligned}
f(\mathcal{S}_i) = g(\mathcal{S}_i) +h(\mathcal{S}_i) ,
\end{aligned}
\end{equation}
where  \( g(\mathcal{S} _i)=\mathrm{\omega}_1\mathcal{U} _{\mathrm{r}}
\) is a modular function and \( h(\mathcal{S}_i) = \mathrm{\omega}_2  \mathcal{U} _{\mathrm{p}}\) is a submodular function. As for $g(\mathcal{S}_i)$, the modular property ensures that the sum of the values for any two subsets is equal to the sum of their union and intersection: \( g(\mathbf{B}_1) + g(\mathbf{B}_2) = g(\mathbf{B}_1 \cup \mathbf{B}_2) + g(\mathbf{B}_1 \cap \mathbf{B}_2) \). As for \( h(\mathcal{S}_i)\), encapsulates the union of perceptual regions. The inherent property of the union operation ensures diminishing returns; adding more regions results in a lesser incremental gain, thereby making \( h \) submodular. Therefore, for any \( \mathbf{B}_1, \mathbf{B}_2 \subseteq \mathcal{S} \): \(h(\mathbf{B}_1) + h(\mathbf{B}_2) \geq h(\mathbf{B}_1 \cup \mathbf{B}_2)+h(\mathbf{B}_1 \cap \mathbf{B}_2)   \). Combining the modular and submodular components, we have \( f(\mathbf{B}_1) + f(\mathbf{B}_2) \geq f(\mathbf{B}_1 \cup \mathbf{B}_2) + f(\mathbf{B}_1 \cap \mathbf{B}_2) \).

(4) To show the fourth property of Definition \ref{def:Submodularity}, we need to prove $\Delta _f(e|\mathbf{B}_1)\geqslant \Delta _f(e|\mathbf{B}_2)$, where $e$ is an arbitrary element in the set $\mathbf{E} \backslash \mathbf{B_2}$. The original inequality can be formulated as:
\begin{equation}\label{eq:submodular4-33}
\begin{aligned}
f(\mathbf{B}_2\cup \{e\})-f(\mathbf{B}_1\cup \{e\})\leqslant f(\mathbf{B}_2)-f(\mathbf{B}_1).
\end{aligned}
\end{equation}
Combining Eqs. (\ref{eq:Monotonicity1}), (\ref{eq:Monotonicity2}), (\ref{eq:submodular4-32}) and (\ref{eq:submodular4-33}), we extend the left side of the equality, which can be given by:
\begin{equation}\label{eq:submodular4-2}
\begin{aligned}
&\mathcal{U} _{\mathrm{sum}}(\mathbf{B}_1\cap \{e\},\mathcal{D})-\mathcal{U} _{\mathrm{sum}}(\mathbf{B}_2\cap \{e\},\mathcal{D}) \\
&= g\left( \mathbf{B}_2\backslash \mathbf{B}_1 \right) +\left[ h\left( \mathbf{B}_2\cap \{e\} \right) -h\left( \mathbf{B}_1\cap \{e\} \right) \right],
\end{aligned}
\end{equation}
Given that $g(\mathbf{B})$ is modular and $h(\mathbf{B})$ is submodular, it follows that \(g\left( \mathbf{B}_2\backslash \mathbf{B}_1 \right) =g\left( \mathbf{B}_2 \right) -g\left( \mathbf{B}_1 \right) \) and $h\left( \mathbf{B}_2\cap \{e\} \right) -h\left( \mathbf{B}_1\cap \{e\} \right) \leqslant h\left( \mathbf{B}_2 \right) -h\left( \mathbf{B}_1 \right) $ for all possible link establishment $e$. Therefore, we can obtain \(\mathcal{U} _{\mathrm{sum}}(\mathbf{B}_1\cap \{e\},\mathcal{D} )-\mathcal{U} _{\mathrm{sum}}(\mathbf{B}_2\cap \{e\},\mathcal{D} )\leqslant g\left( \mathbf{B}_2 \right) -g\left( \mathbf{B}_1 \right) +h\left( \mathbf{B}_2 \right) -h\left( \mathbf{B}_1 \right) =\mathcal{U} _{\mathrm{sum}}(\mathbf{B}_2,\mathcal{D} )-\mathcal{U} _{\mathrm{sum}}(\mathbf{B}_1,\mathcal{D} )\). {\hfill $\blacksquare$\par}
\end{Proof}\par

\section{Network Structure Details}\label{appendix: Network Structure Details}
\subsection{Encoder Phase}
The input image $\mathbf{x}$ undergoes transformation to latent representation $\mathbf{z}$. A quantization function ${Q}$ further refines $\mathbf{z}$ into a discrete vector $q\in \mathbb{H}^{D}$, leading to $q = Q(\mathbf{z})$. In V2V scenarios, $q$ is serialized into bitstream $b$, with entropy coding to maximize the bandwidth efficiency. As shown in Fig. \ref{fig:dl_architecture}(a), the encoder compresses the input camera data, reducing its size for transmission while retaining essential information. It comprises the following layers:
\begin{itemize}
    \item Convolutional Layers (Conv): There are three convolutional layers with kernel sizes of \(9 \times 9\), \(5 \times 5\), and \(5 \times 5\), and strides of 4, 2, and 2, respectively. These layers progressively reduce the spatial dimensions of the input data from \([H, W, 3]\) to \([H/16, W/16, F3]\), where \(H\) is the height, \(W\) is the width, and \(F1\), \(F2\), and \(F3\) are the number of filters at each layer.
    \item Generalized Divisive Normalization (GDN) Layers: GDN layers follow each convolutional layer to normalize the feature maps, promoting sparsity and reducing redundancy, which is critical for efficient compression. The output dimensions remain the same as those of the preceding convolutional layers.
    \item Element-wise Modulation: The outputs of the GDN layers are modulated by element-wise multiplication with modulation vectors \(\mathbf{m}(h_{ij}) = (m_1(h_{ij}), m_2(h_{ij}), m_3(h_{ij}))\) generated by the modulation network. This modulation adapts the compression process to the CSI feedback, \(h_{ij}\).
    \item Quantization: The modulated data is then quantized, reducing the precision of the latent representation and enabling efficient encoding into a bitstream.
\end{itemize}

\begin{figure*}
    \centering
    \subfigure[The encoder part.]{
        \includegraphics[width=9cm]{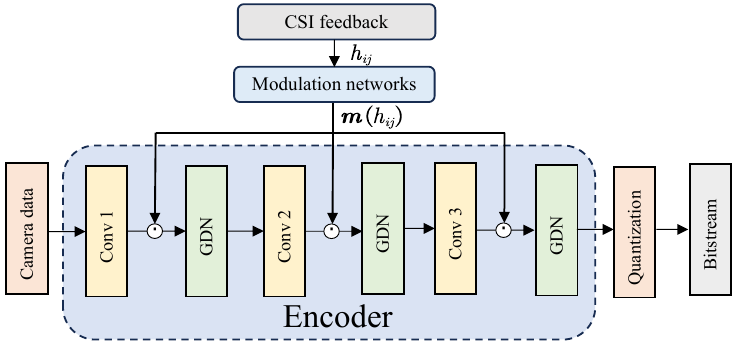}\label{fig:encoder4}
    }
    \subfigure[The decoder part.]{
        \includegraphics[width=8cm]{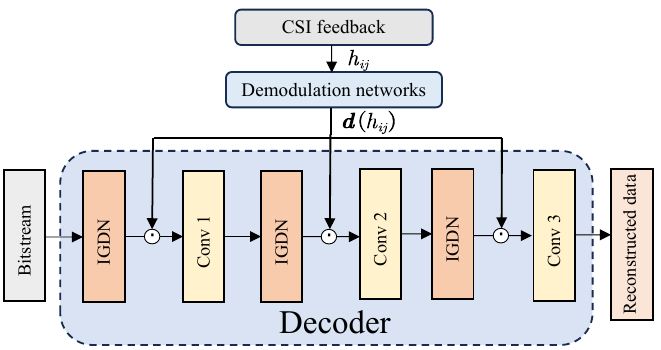}\label{fig:decoder4}
    }
    \caption{{\color{black}The architecture of the deep learning-based adaptive compression.}}
    \label{fig:dl_architecture}
\end{figure*}

\subsection{Decoding Phase}
The bitstream $b$ undergoes inverse transformations. This includes dequantization ($\hat{\mathbf{z}} = Q^{-1}(q)$) and a reconstruction function ($\hat{\mathbf{x}} = H({\hat{\mathbf{z}}})$) to retrieve the original image. As shown in Fig. \ref{fig:dl_architecture}(b), the decoder reconstructs the original camera data from the compressed bitstream, ensuring that the essential information is preserved and accurately recovered. It consists of:
\begin{itemize}
    \item Inverse Generalized Divisive Normalization (IGDN) Layers: IGDN layers, which reverse the normalization effect of the GDN layers, are applied to the quantized latent representation. The output dimensions are \([H/16, W/16, F3]\).
    \item Element-wise Demodulation: The outputs of the IGDN layers are demodulated by element-wise multiplication with demodulation vectors \(\mathbf{d}(h_{ij}) = (d_1(h_{ij}), d_2(h_{ij}), d_3(h_{ij}))\) generated by the demodulation network, adapting the reconstruction to the CSI feedback.
    \item Deconvolutional Layers (Deconv): These layers upsample the feature maps to reconstruct the data, with kernel sizes of \(5 \times 5\), \(5 \times 5\), and \(9 \times 9\), and strides of 2, 2, and 4, respectively, restoring the dimensions to \([H, W, 3]\).
\end{itemize}
\begin{figure}[t]
  \centering
  \includegraphics[width=0.42\textwidth]{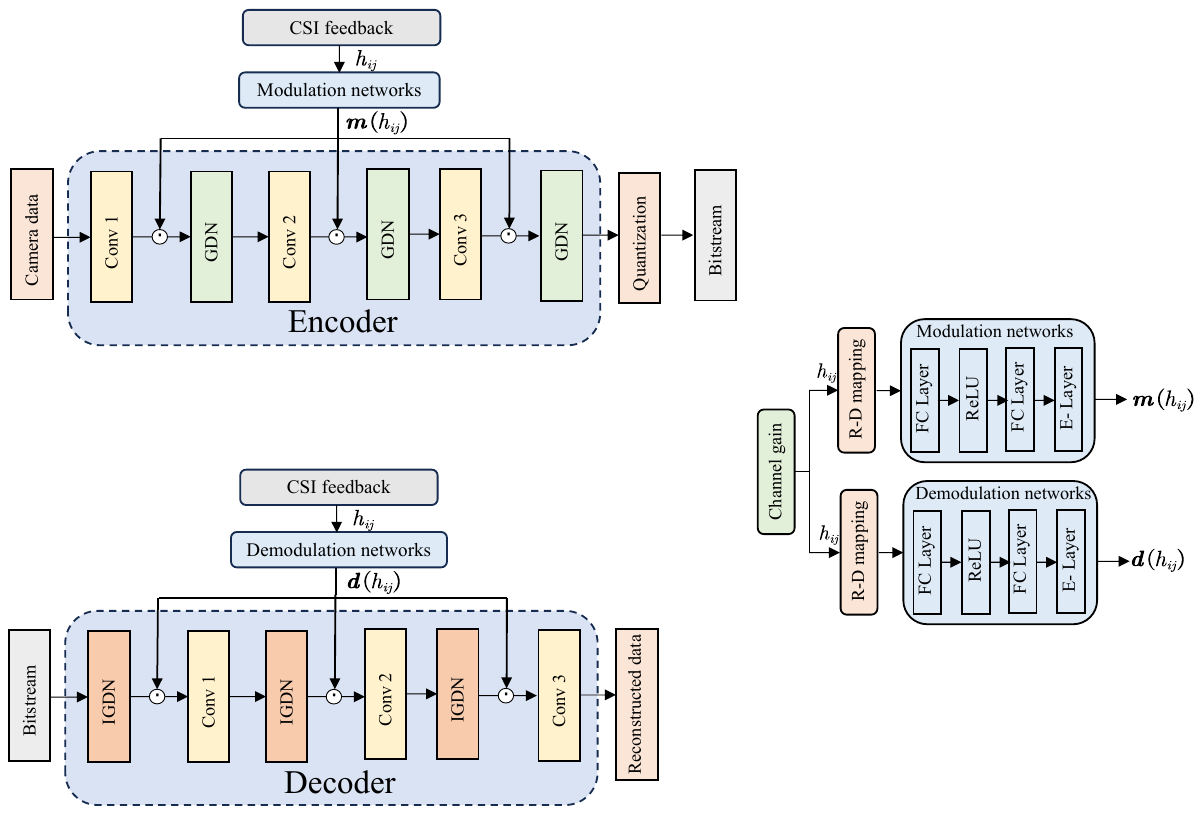}
  \caption{{\color{black}The architecture of modulation/ demodulation networks.}}
  \label{fig:module}
\end{figure}

\subsection{Modulation/Demodulation Networks}
As shown in Fig. \ref{fig:module}, the modulation and demodulation networks translate the CSI feedback into adaptive modulation and demodulation vectors, respectively, for the encoder and decoder. They include:
\begin{itemize}
    \item R-D Mapping: Converts the CSI feedback \(h_{ij}\) into rate-distortion lambda parameters.
    \item Fully Connected (FC) Layers: Multiple FC layers with ReLU activations introduce non-linearity and enhance the representation capacity of the network.
    \item Exponential Layer (E-Layer): Ensures that the output vectors are positive, necessary for the subsequent modulation and demodulation processes.
\end{itemize}

\subsection{Fine-tuning Phase}
This adaptive approach enables the ego CAV and nearby CAVs to adaptively adjust the data compression ratio according to dynamic channel conditions. As shown in Fig. \ref{fig: adaptive compression}, we introduce a fine-tuning compression strategy to further reduce temporal redundancy between frames. Utilizing a Modulated Autoencoder (MAE), our method first transmits a few uncompressed images to RSU for real-time fine-tuning.

Even though digitization and quantization modules are present in the encoder and decoder, we employ scalar quantization by rounding to the nearest neighbor during inference. During training, this is replaced by additive uniform noise as a proxy, i.e., \( \mathbf{\tilde{z}} = \mathbf{z} + \Delta\mathbf{z} \) with \( \Delta\mathbf{z} \sim \mathcal{U}(-\frac{1}{2}, \frac{1}{2}) \). This approach allows the network to approximate gradients during the backward pass, enabling effective backpropagation through these non-differentiable layers. This technique ensures that the network can be trained end-to-end, despite the presence of quantization layers.

\subsection{The Evaluation of Rate-Distoration (R-D) Performance}\label{appendix: The Evaluation of Rate-Distoration (R-D) Performance}
\begin{figure}[t]
  \centering 
  \subfigure[{\color{black}MS-SSIM vs. bit per pixel levels}]{
  \begin{minipage}[t]{0.5\linewidth}
    \centering 
  \includegraphics[width=1.65in]{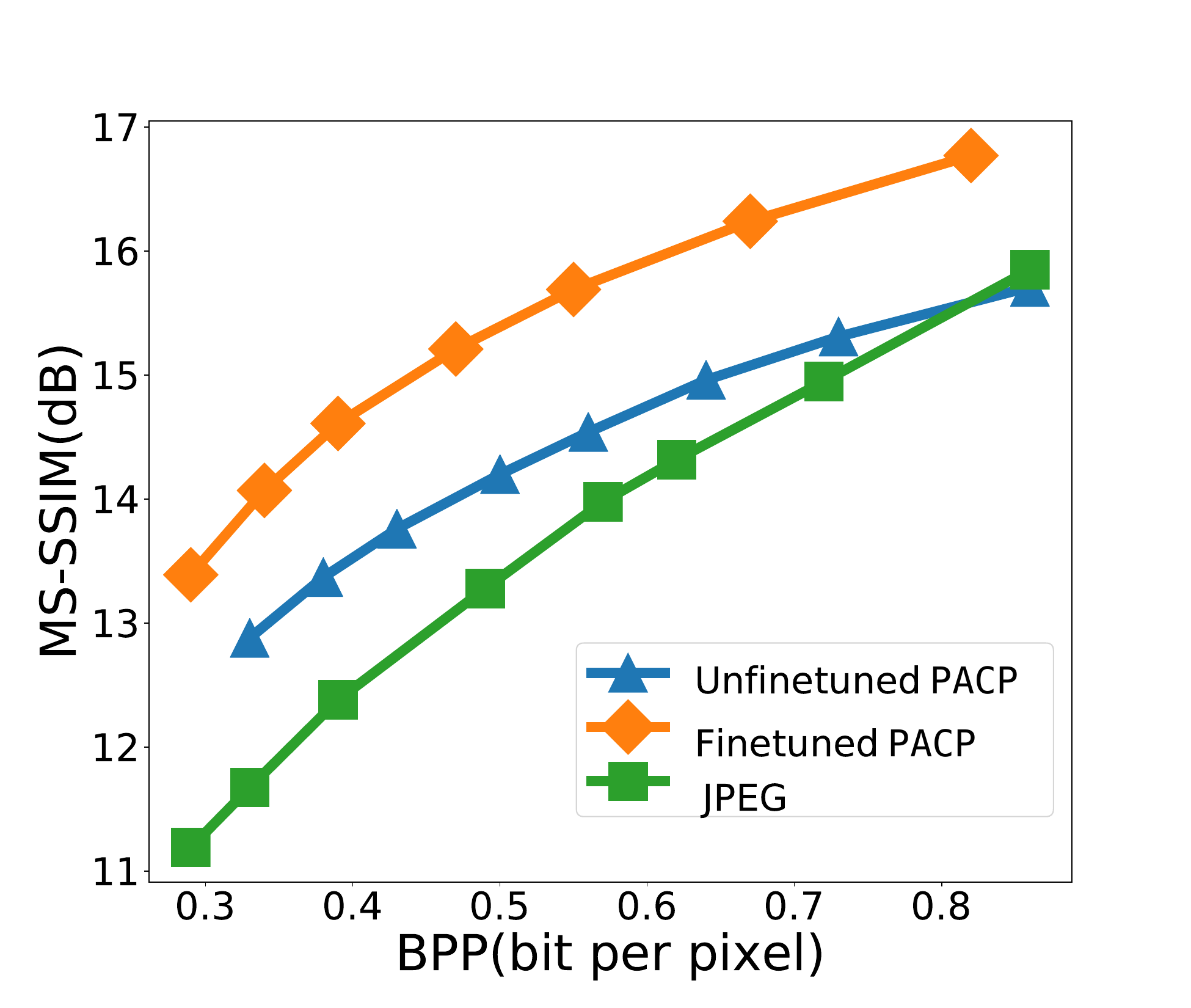}\label{fig:ms-ssim}
  \end{minipage}%
  }%
  \subfigure[{\color{black}PSNR vs. bit per pixel levels}]{
  \begin{minipage}[t]{0.50\linewidth}
    \centering 
  \includegraphics[width=1.65in]{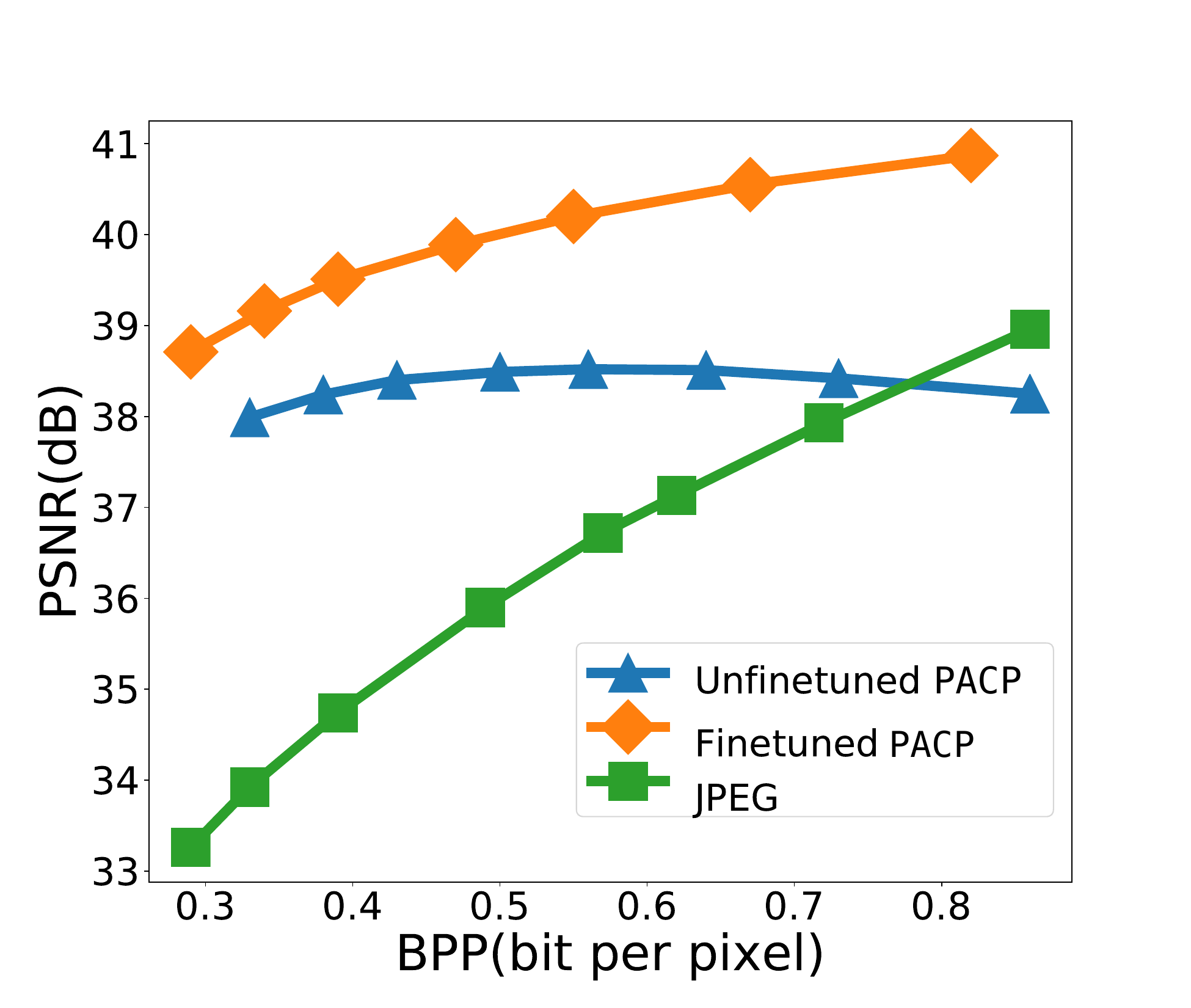}\label{fig:psnr}
  \end{minipage}%
  }%
  \flushleft 
  \caption{{\color{black}Comparison of rate–distortion (R-D) performance between finetuned PACP, unfinetuned PACP and JPEG at the different bit per pixel levels.}}
  \label{fig:bandwidth saving}
\end{figure}

{\color{black}We evaluate the R-D performance of our proposed PACP framework using the Multi-scale Structural Similarity (MS-SSIM) and Peak Signal-to-Noise Ratio (PSNR) metrics. MS-SSIM is a perceptual metric that measures the similarity between the original and compressed images, accounting for variations in resolution and viewing conditions. Besides, PSNR quantifies the reconstruction quality of images by comparing the original and compressed versions, with higher values indicating better quality. As shown in Fig. \ref{fig:ms-ssim}, at lower BPP levels (below 0.4), our PACP method, both unfine-tuned and fine-tuned, achieves at least 1-2 dB higher MS-SSIM compared with JPEG. At higher BPP levels, the fine-tuned PACP maintains at least a 1 dB advantage. Fig. \ref{fig:psnr} demonstrates that the fine-tuned PACP consistently improves PSNR over JPEG by 2-5 dB across all BPP levels. These results validate the superior R-D performance of our deep learning-based adaptive compression method compared with the standard JPEG compression.
}

\end{appendices}

\ifCLASSOPTIONcaptionsoff
  \newpage
\fi

\bibliographystyle{IEEEtran}
\bibliography{ref}

\begin{thebibliography}{10}
\providecommand{\url}[1]{#1}
\csname url@samestyle\endcsname
\providecommand{\newblock}{\relax}
\providecommand{\bibinfo}[2]{#2}
\providecommand{\BIBentrySTDinterwordspacing}{\spaceskip=0pt\relax}
\providecommand{\BIBentryALTinterwordstretchfactor}{4}
\providecommand{\BIBentryALTinterwordspacing}{\spaceskip=\fontdimen2\font plus
\BIBentryALTinterwordstretchfactor\fontdimen3\font minus \fontdimen4\font\relax}
\providecommand{\BIBforeignlanguage}[2]{{%
\expandafter\ifx\csname l@#1\endcsname\relax
\typeout{** WARNING: IEEEtran.bst: No hyphenation pattern has been}%
\typeout{** loaded for the language `#1'. Using the pattern for}%
\typeout{** the default language instead.}%
\else
\language=\csname l@#1\endcsname
\fi
#2}}
\providecommand{\BIBdecl}{\relax}
\BIBdecl

\bibitem{hu2024magazine}
S.~Hu, Z.~Fang, Y.~Deng, X.~Chen, and Y.~Fang, ``Collaborative perception for connected and autonomous driving: {C}hallenges, possible solutions and opportunities,'' \emph{arXiv preprint arXiv:2401.01544}, 2024.

\bibitem{hu2023adaptive}
S.~Hu, Z.~Fang, H.~An, G.~Xu, Y.~Zhou, X.~Chen, and Y.~Fang, ``Adaptive communications in collaborative perception with domain alignment for autonomous driving,'' \emph{arXiv preprint arXiv:2310.00013}, 2023.

\bibitem{hu2023towards}
S.~Hu, Z.~Fang, X.~Chen, Y.~Fang, and S.~Kwong, ``Towards full-scene domain generalization in multi-agent collaborative bird's eye view segmentation for connected and autonomous driving,'' \emph{arXiv preprint arXiv:2311.16754}, 2023.

\bibitem{chen2023vehicle}
X.~Chen, Y.~Deng, H.~Ding, G.~Qu, H.~Zhang, P.~Li, and Y.~Fang, ``Vehicle as a service ({VaaS}): {L}everage vehicles to build service networks and capabilities for smart cities,'' \emph{arXiv preprint arXiv:2304.11397}, 2023.

\bibitem{10415235}
Z.~Lin, G.~Zhu, Y.~Deng, X.~Chen, Y.~Gao, K.~Huang, and Y.~Fang, ``Efficient parallel split learning over resource-constrained wireless edge networks,'' \emph{IEEE Transactions on Mobile Computing}, pp. 1--16, Jan. 2024.

\bibitem{9165167}
Y.~Xiao, F.~Codevilla, A.~Gurram, O.~Urfalioglu, and A.~M. López, ``Multimodal end-to-end autonomous driving,'' \emph{IEEE Transactions on Intelligent Transportation Systems}, vol.~23, no.~1, pp. 537--547, Aug. 2022.

\bibitem{xu2022cobevt}
R.~Xu, Z.~Tu, H.~Xiang, W.~Shao, B.~Zhou, and J.~Ma, ``{CoBEVT}: {C}ooperative bird's eye view semantic segmentation with sparse transformers,'' in \emph{Conference on Robot Learning (CoRL)}, Auckland, NZ, Dec. 2022.

\bibitem{zhang2024smartcooper}
Y.~Zhang, H.~An, Z.~Fang, G.~Xu, Y.~Zhou, X.~Chen, and Y.~Fang, ``{SmartCooper: Vehicular} collaborative perception with adaptive fusion and judger mechanism,'' in \emph{IEEE International Conference on Robotics and Automation (ICRA)}, Yokohama, Japan, May 2024.

\bibitem{9682601}
Y.~Yuan, H.~Cheng, and M.~Sester, ``Keypoints-based deep feature fusion for cooperative vehicle detection of autonomous driving,'' \emph{IEEE Robotics and Automation Letters}, vol.~7, no.~2, pp. 3054--3061, Jan. 2022.

\bibitem{9779322}
M.~Noor-A-Rahim, Z.~Liu, H.~Lee, M.~O. Khyam, J.~He, D.~Pesch, K.~Moessner, W.~Saad, and H.~V. Poor, ``{6G} for vehicle-to-everything {(V2X)} communications: {Enabling} technologies, challenges, and opportunities,'' \emph{Proceedings of the IEEE}, vol. 110, no.~6, pp. 712--734, May 2022.

\bibitem{liu2020when2com}
Y.-C. Liu, J.~Tian, N.~Glaser, and Z.~Kira, ``When2com: {M}ulti-agent perception via communication graph grouping,'' in \emph{IEEE/CVF Conference on Computer Vision and Pattern Recognition (CVPR)}, Seattle, WA (Virtual), 2020, pp. 4106--4115.

\bibitem{10193767}
Z.~Wang, Z.~Zhang, J.~Wang, C.~Jiang, W.~Wei, and Y.~Ren, ``{AUV}-assisted node repair for {IoUT} relying on multi-agent reinforcement learning,'' \emph{IEEE Internet of Things Journal, (DOI: 10.1109/JIOT.2023.3298522)}, 2023.

\bibitem{9451536}
Z.~Fang, J.~Wang, J.~Du, X.~Hou, Y.~Ren, and Z.~Han, ``Stochastic optimization-aided energy-efficient information collection in internet of underwater things networks,'' \emph{IEEE Internet of Things Journal}, vol.~9, no.~3, pp. 1775--1789, Jun. 2022.

\bibitem{10517492}
Z.~Wang, J.~Du, C.~Jiang, Y.~Ren, and X.-P. Zhang, ``{UAV}-assisted target tracking and computation offloading in {USV}-based {MEC} networks,'' \emph{IEEE Transactions on Mobile Computing, (DOI: 10.1109/TMC.2024.3396121)}, pp. 1--16, 2024.

\bibitem{6248074}
A.~Geiger, P.~Lenz, and R.~Urtasun, ``Are we ready for autonomous driving? {The KITTI vision} benchmark suite,'' in \emph{IEEE Conference on Computer Vision and Pattern Recognition (CVPR)}, Providence, RI, Jul. 2012, pp. 3354--3361.

\bibitem{chen2019f}
Q.~Chen, X.~Ma, S.~Tang, J.~Guo, Q.~Yang, and S.~Fu, ``F-cooper: {F}eature based cooperative perception for autonomous vehicle edge computing system using {3D} point clouds,'' in \emph{Proceedings of the 4th ACM/IEEE Symposium on Edge Computing}, Arlington, Virginia, Nov. 2019, pp. 88--100.

\bibitem{wang2020v2vnet}
T.-H. Wang, S.~Manivasagam, M.~Liang, B.~Yang, W.~Zeng, and R.~Urtasun, ``{V2VNet: {V}ehicle-to-vehicle communication for joint perception and prediction},'' in \emph{European Conference on Computer Vision (ECCV)}, Glasgow, United Kingdom, Aug. 2020, pp. 605--621.

\bibitem{li2023learning}
J.~Li, R.~Xu, X.~Liu, J.~Ma, Z.~Chi, J.~Ma, and H.~Yu, ``Learning for vehicle-to-vehicle cooperative perception under lossy communication,'' \emph{IEEE Transactions on Intelligent Vehicles}, vol.~8, no.~4, pp. 2650--2660, Mar. 2023.

\bibitem{9718315}
S.~Wang, Y.~Hong, R.~Wang, Q.~Hao, Y.-C. Wu, and D.~W.~K. Ng, ``Edge federated learning via unit-modulus over-the-air computation,'' \emph{IEEE Transactions on Communications}, vol.~70, no.~5, pp. 3141--3156, Feb. 2022.

\bibitem{skodras2001jpeg}
A.~Skodras, C.~Christopoulos, and T.~Ebrahimi, ``The {JPEG} 2000 still image compression standard,'' \emph{IEEE Signal Processing Magazine}, vol.~18, no.~5, pp. 36--58, Sep. 2001.

\bibitem{8945224}
H.~Liu, H.~Yuan, Q.~Liu, J.~Hou, and J.~Liu, ``A comprehensive study and comparison of core technologies for {MPEG 3-D} point cloud compression,'' \emph{IEEE Transactions on Broadcasting}, vol.~66, no.~3, pp. 701--717, Dec. 2020.

\bibitem{chen2019cooper}
Q.~Chen, S.~Tang, Q.~Yang, and S.~Fu, ``Cooper: {C}ooperative perception for connected autonomous vehicles based on 3d point clouds,'' in \emph{IEEE 39th International Conference on Distributed Computing Systems (ICDCS)}, Dallas, TX, Oct. 2019, pp. 514--524.

\bibitem{rawashdeh2018collaborative}
Z.~Y. Rawashdeh and Z.~Wang, ``Collaborative automated driving: {A} machine learning-based method to enhance the accuracy of shared information,'' in \emph{21st International Conference on Intelligent Transportation Systems (ITSC)}, Maui, Hawaii, Dec. 2018, pp. 3961--3966.

\bibitem{lyu2022distributed}
X.~Lyu, C.~Zhang, C.~Ren, and Y.~Hou, ``Distributed graph-based optimization of multicast data dissemination for internet of vehicles,'' \emph{IEEE Transactions on Intelligent Transportation Systems}, vol.~24, no.~3, pp. 3117--3128, Dec. 2022.

\bibitem{9204672}
B.~L. Nguyen, D.~T. Ngo, N.~H. Tran, M.~N. Dao, and H.~L. Vu, ``Dynamic {V2I/V2V} cooperative scheme for connectivity and throughput enhancement,'' \emph{IEEE Transactions on Intelligent Transportation Systems}, vol.~23, no.~2, pp. 1236--1246, Sep. 2022.

\bibitem{8812911}
Y.~Ma, W.~Liang, J.~Wu, and Z.~Xu, ``Throughput maximization of {NFV}-enabled multicasting in mobile edge cloud networks,'' \emph{IEEE Transactions on Parallel and Distributed Systems}, vol.~31, no.~2, pp. 393--407, 2020.

\bibitem{zhang2018vehicular}
S.~Zhang, J.~Chen, F.~Lyu, N.~Cheng, W.~Shi, and X.~Shen, ``Vehicular communication networks in the automated driving era,'' \emph{IEEE Communications Magazine}, vol.~56, no.~9, pp. 26--32, Sep. 2018.

\bibitem{yang2020variable}
F.~Yang, L.~Herranz, J.~Van De~Weijer, J.~A.~I. Guiti{\'a}n, A.~M. L{\'o}pez, and M.~G. Mozerov, ``Variable rate deep image compression with modulated autoencoder,'' \emph{IEEE Signal Processing Letters}, vol.~27, pp. 331--335, Jul. 2020.

\bibitem{liao2024bat}
H.~Liao, Z.~Li, H.~Shen, W.~Zeng, D.~Liao, G.~Li, and C.~Xu, ``{BAT: B}ehavior-aware human-like trajectory prediction for autonomous driving,'' in \emph{Proceedings of the AAAI Conference on Artificial Intelligence (AAAI)}, Vancouver, Canada, Feb. 2024, pp. 10\,332--10\,340.

\bibitem{10262233}
Z.~Wen, R.~Yang, B.~Qian, Y.~Xuan, L.~Lu, Z.~Wang, H.~Peng, J.~Xu, A.~Y. Zomaya, and R.~Ranjan, ``{JANUS: L}atency-aware traffic scheduling for {IoT} data streaming in edge environments,'' \emph{IEEE Transactions on Services Computing}, vol.~16, no.~6, pp. 4302--4316, Sep. 2023.

\bibitem{10327699}
H.~Cao, J.~Xu, Z.~Yang, L.~Shangguan, J.~Zhang, X.~He, and Y.~Liu, ``Scaling up edge-assisted real-time collaborative visual {SLAM} applications,'' \emph{IEEE/ACM Transactions on Networking}, vol.~32, no.~2, pp. 1823--1838, Nov. 2024.

\bibitem{an2024throughput}
H.~An, Z.~Fang, Y.~Zhang, S.~Hu, X.~Chen, G.~Xu, and Y.~Fang, ``Cooperative vehicular perception with throughput maximization and adaptive compression,'' in \emph{in the submission to IEEE International Conference on Distributed Computing Systems (ICDCS)}, 2024, password: FJUfPp! [Online]. Available: \url{https://zhengrufang.com/paper/ICDCS.pdf}.

\bibitem{9681261}
E.~R. Magsino and I.~W.-H. Ho, ``An enhanced information sharing roadside unit allocation scheme for vehicular networks,'' \emph{IEEE Transactions on Intelligent Transportation Systems}, vol.~23, no.~9, pp. 15\,462--15\,475, Jan. 2022.

\bibitem{dosovitskiy2017carla}
A.~Dosovitskiy, G.~Ros, F.~Codevilla, A.~Lopez, and V.~Koltun, ``{CARLA: A}n open urban driving simulator,'' in \emph{Conference on robot learning (CoRL)}, Mountain View, California, Nov. 2017, pp. 1--16.

\bibitem{xu2022opv2v}
R.~Xu, H.~Xiang, X.~Xia, X.~Han, J.~Li, and J.~Ma, ``{OPV2V}: {An} open benchmark dataset and fusion pipeline for perception with vehicle-to-vehicle communication,'' in \emph{International Conference on Robotics and Automation (ICRA)}, Philadelphia, PA, May 2022, pp. 2583--2589.

\bibitem{9403386}
M.~Adil, H.~Song, J.~Ali, M.~A. Jan, M.~Attique, S.~Abbas, and A.~Farouk, ``Enhanced-{AODV}: {A} robust three phase priority-based traffic load balancing scheme for internet of things,'' \emph{IEEE Internet of Things Journal}, vol.~9, no.~16, pp. 14\,426--14\,437, Apr. 2022.

\bibitem{luo2022learning}
X.~Luo and P.~Li, ``Learning-based off-chain transaction scheduling in prioritized payment channel networks,'' \emph{IEEE Journal on Selected Areas in Communications}, vol.~40, no.~12, pp. 3589--3599, Oct. 2022.

\bibitem{8891313}
W.~Anwar, N.~Franchi, and G.~Fettweis, ``Physical layer evaluation of {V2X} communications technologies: {5G NR-V2X, LTE-V2X, IEEE 802.11bd, and IEEE 802.11p},'' in \emph{IEEE 90th Vehicular Technology Conference (VTC2019-Fall)}, Honolulu, HI, Sep. 2019, pp. 1--7.

\bibitem{adil2021enhanced}
M.~Adil, H.~Song, J.~Ali, M.~A. Jan, M.~Attique, S.~Abbas, and A.~Farouk, ``Enhanced-{AODV}: {A} robust three phase priority-based traffic load balancing scheme for internet of things,'' \emph{IEEE Internet of Things Journal}, vol.~9, no.~16, pp. 14\,426--14\,437, Apr. 2021.

\bibitem{9849675}
Z.~Xiao, Z.~Han, A.~Nallanathan, O.~A. Dobre, B.~Clerckx, J.~Choi, C.~He, and W.~Tong, ``Antenna array enabled space/air/ground communications and networking for {6G},'' \emph{IEEE Journal on Selected Areas in Communications}, vol.~40, no.~10, pp. 2773--2804, Aug. 2022.

\bibitem{10158439}
W.~Chen, X.~Lin, J.~Lee, A.~Toskala, S.~Sun, C.~F. Chiasserini, and L.~Liu, ``{5G-Advanced} toward {6G}: {P}ast, present, and future,'' \emph{IEEE Journal on Selected Areas in Communications}, vol.~41, no.~6, pp. 1592--1619, Jun. 2023.

\bibitem{dobzinski2013communication}
S.~Dobzinski and J.~Vondr{\'a}k, ``Communication complexity of combinatorial auctions with submodular valuations,'' in \emph{Twenty-Fourth Annual ACM-SIAM Symposium on Discrete Algorithms}, 2013, pp. 1205--1215.

\bibitem{krause2014submodular}
A.~Krause and D.~Golovin, ``Submodular function maximization.'' \emph{Tractability}, vol.~3, no. 71-104, p.~3, 2014.

\bibitem{balle2016end}
J.~Ball{\'e}, V.~Laparra, and E.~P. Simoncelli, ``End-to-end optimized image compression,'' in \emph{International Conference on Learning Representations (ICLR)}, Toulon, France, Apr. 2017, pp. 1--27.

\end{thebibliography}

\begin{IEEEbiographynophoto}{Zhengru Fang}
(S'20) received his B.S. degree in electronics and information engineering from Huazhong University of Science and Technology, Wuhan, China, in 2019, and his M.S. degree from Tsinghua University, Beijing, China, in 2022. He is currently pursuing a Ph.D. in the Department of Computer Science at City University of Hong Kong. His research interests include collaborative perception, V2X, age of information, and mobile edge computing. He serves as a reviewer for ACM Computing Surveys, IEEE TMC, IEEE JSAC, IEEE TIV, IEEE IoTJ, IEEE TVT, and IEEE VTM.
\end{IEEEbiographynophoto}

\begin{IEEEbiographynophoto}{Senkang Hu}
received his B.S. degree in electronic and information engineering from Beijing Institute of Technology, Beijing, China, in 2022. He is currently pursuing his PhD degree in the Department of Computer Science at City University of Hong Kong, Hong Kong. His research interests include autonomous driving, vehicle-to-vehicle collaborative perception.
\end{IEEEbiographynophoto}
\vspace{-3mm}

\begin{IEEEbiographynophoto}{Haonan An} received his B.S. degree in Telecommunication Engineering from Huazhong University of Science and Technology, Wuhan, China, and his M.S. degree from the School of Electrical and Electronic Engineering, Nanyang Technological University, Singapore. He is currently pursuing a Ph.D. in the Department of Computer Science at City University of Hong Kong. His research interests are focused on AI security, generative models, and collaborative perception.
\end{IEEEbiographynophoto}

\begin{IEEEbiographynophoto}{Yuang Zhang}
         is currently a Ph. D. student in the Department of Civil and Environmental Engineering, University of Washington, Seattle, WA, USA. He obtained the Bachelor and Master's degree in the Department of Automation, Tsinghua University. Mr. Zhang has published papers in IEEE ICRA, IEEE IST and COTA ICTP, etc. His research interests lie in the area of Intelligent transportation.
\end{IEEEbiographynophoto}

\begin{IEEEbiographynophoto}{Jingjing Wang} (S'14-M'19-SM'21) received his B.S. degree in Electronic Information Engineering from Dalian University of Technology, Liaoning, China in 2014 and the Ph.D. degree in Information and Communication Engineering from Tsinghua University, Beijing, China in 2019, both with the highest honors. From 2017 to 2018, he visited the Next Generation Wireless Group chaired by Prof. Lajos Hanzo, University of Southampton, UK. Dr. Wang is currently an associate professor at School of Cyber Science and Technology, Beihang University, and also a researcher at Hangzhou Innovation Institute, Beihang University, Hangzhou, China. His research interests include AI enhanced next-generation wireless networks, swarm intelligence and confrontation. He has published over 100 IEEE Journal/Conference papers. Dr. Wang was a recipient of the Best Journal Paper Award of IEEE ComSoc Technical Committee on Green Communications \& Computing in 2018, the Best Paper Award of IEEE ICC and IWCMC in 2019.
\end{IEEEbiographynophoto}

\begin{IEEEbiographynophoto}{Hangcheng Cao} is currently a postdoctoral fellow in the Department of Computer Science, City University of Hong Kong. He obtained the Ph.D. degree in the College of Computer Science and Electronic Engineer, Hunan University, China, in 2023. He studied as a joint PhD student in the School of Computer Science and Engineering, Nanyang Technological University, Singapore, in 2022. He has published papers in IEEE S\&P, ACM Ubicomp/IMWUT, IEEE INFOCOM, IEEE TMC, IEEE ICDCS, ACM MobiCom Workshop, ACM ToSN, IEEE Communications Magazine, etc. His research interests lie in the area of IoT security.
\end{IEEEbiographynophoto}

\begin{IEEEbiographynophoto}{Xianhao Chen}
        received the B.Eng. degree in electronic information from Southwest Jiaotong University in 2017, and the Ph.D. degree in electrical and computer engineering from the University of Florida in 2022. He is currently an assistant professor with the Department of Electrical and Electronic Engineering, the University of Hong Kong. He received the 2022 ECE graduate excellence award for research from the University of Florida. His research interests include wireless networking, edge intelligence, and machine learning.
\end{IEEEbiographynophoto}

\begin{IEEEbiographynophoto}{Yuguang Fang}
(S’92, M’97, SM’99, F’08) received the MS degree from Qufu Normal University, China in 1987, a PhD degree from Case Western Reserve University, Cleveland, Ohio, USA, in 1994, and a PhD degree from Boston University, Boston, Massachusetts, USA in 1997. He joined the Department of Electrical and Computer Engineering at University of Florida in 2000 as an assistant professor, then was promoted to associate professor in 2003, full professor in 2005, and distinguished professor in 2019, respectively. Since August 2022, he has been the Chair Professor of Internet of Things with the Department of Computer Science at City University of Hong Kong. Prof. Fang received many awards including the US NSF CAREER Award (2001), US ONR Young Investigator Award (2002), 2018 IEEE Vehicular Technology Outstanding Service Award, IEEE Communications Society AHSN Technical Achievement Award (2019), CISTC Technical Recognition Award (2015), and WTC Recognition Award (2014), and 2010-2011 UF Doctoral Dissertation Advisor/Mentoring Award. He held multiple professorships including the Changjiang Scholar Chair Professorship (2008-2011), Tsinghua University Guest Chair Professorship (2009-2012), University of Florida Foundation Preeminence Term Professorship (2019-2022), and University of Florida Research Foundation Professorship (2017-2020, 2006- 2009). He served as the Editor-in-Chief of IEEE Transactions on Vehicular Technology (2013-2017) and IEEE Wireless Communications (2009-2012) and serves/served on several editorial boards of journals including Proceedings of the IEEE (2018-present), ACM Computing Surveys (2017-present), ACM Transactions on Cyber-Physical Systems (2020-present), IEEE Transactions on Mobile Computing (2003-2008, 2011-2016, 2019-present), IEEE Transactions on Communications (2000-2011), and IEEE Transactions on Wireless Communications (2002-2009). He served as the Technical Program CoChair of IEEE INFOCOM’2014. He is a Member-at-Large of the Board of Governors of IEEE Communications Society (2022-2024) and the Director of Magazines of IEEE Communications Society (2018-2019). He is a fellow of ACM, IEEE, and AAAS.
\end{IEEEbiographynophoto}

\end{document}